\documentclass[preprint,superscriptaddress]{revtex4}
\usepackage{graphicx}

\usepackage{amsmath}
\usepackage{amsthm}
\usepackage{amsfonts}
\usepackage{amssymb}

\newcommand{\Hodd}{\hat{H}_{n,\mathrm{odd}}}
\newcommand{\Heven}{\hat{H}_{n,\mathrm{even}}}
\newcommand{\psieven}{\psi_{\mathrm{even}}}
\newcommand{\nnum}{\mathbb{N}}
\newcommand{\cnum}{\mathbb{C}}
\newcommand{\e}{\mathrm{e}}
\newcommand{\hilbert}{\mathcal{H}}

\newtheorem{theorem}{Theorem}[section]
\newtheorem{proposition}[theorem]{Proposition}\theoremstyle{remark}

\usepackage{xcolor}

\newcommand{\bra}[1] {\left\langle #1 \right|}
\newcommand{\ket}[1] {\left| #1 \right\rangle}
\newcommand{\braket}[2] {\left\langle #1 | #2 \right\rangle}
\newcommand{\domain}{\mathcal{D}}

\usepackage{tikz}
\usetikzlibrary{arrows.meta}

\begin{document}

\title{Finite-dimensional approximations of generalized squeezing}

\author{Sahel Ashhab}
\affiliation{\small Advanced ICT Research Institute, National Institute of Information and Communications Technology (NICT), 4-2-1, Nukui-Kitamachi, Koganei, Tokyo 184-8795, Japan}
\affiliation{\small Research Institute for Science and Technology, Tokyo University of Science, 1-3 Kagurazaka, Shinjuku-ku, Tokyo 162-8601, Japan}

\author{Felix Fischer}
\affiliation{\small Department Physik, Friedrich-Alexander-Universität Erlangen-Nürnberg, Staudtstraße 7, 91058 Erlangen, Germany}

\author{Davide Lonigro}
\affiliation{\small Department Physik, Friedrich-Alexander-Universität Erlangen-Nürnberg, Staudtstraße 7, 91058 Erlangen, Germany}

\author{Daniel Braak}
\affiliation{\small TP III and Center for Electronic Correlations and Magnetism, University of Augsburg, 86135 Augsburg, Germany}

\author{Daniel Burgarth}
\affiliation{\small Department Physik, Friedrich-Alexander-Universität Erlangen-Nürnberg, Staudtstraße 7, 91058 Erlangen, Germany}

\begin{abstract}
We show unexpected behaviour in simulations of generalized squeezing performed with finite-dimensional truncations of the Fock space: even for extremely large dimension of the state space, the results depend on whether the truncation dimension is even or odd. This situation raises the question whether the simulation results are physically meaningful. We demonstrate that, in fact, the two truncation schemes correspond to two well-defined, distinct unitary evolutions whose generators are defined on different subsets of the infinite-dimensional Fock space. This is a consequence of the fact that the generalized squeezing Hamiltonian is not self-adjoint on states with finite excitations, but possesses multiple self-adjoint extensions. Furthermore, we present results on the spectrum of the squeezing Hamiltonians corresponding to even and odd truncation size that elucidate the properties of the two different self-adjoint extensions corresponding to the even and odd truncation scheme. To make the squeezing operator applicable to a physical system, we must regularize it by other terms that depend on the specifics of the experimental implementation. We show that the addition of a Kerr interaction term in the Hamiltonian leads to uniquely converging simulations, with no dependence on the parity of the truncation size, and demonstrate that the Kerr term indeed renders the Hamiltonian self-adjoint and thus physically interpretable.
\end{abstract}

\maketitle

\newpage

\section{Introduction}
\label{Sec:Introduction}

There has been a rapid advance in the development of quantum optical technology in recent years, including the realization of nonlinear effects~\cite{Walls,Scully,Drummond}. Physical systems that realize multiphoton interactions have been proposed theoretically~\cite{Felicetti,Ayyash2024} and demonstrated experimentally~\cite{Chang,Eriksson,Bazavan,Saner}. One of the applications of these systems is generalized squeezing, in which photons are generated in multiples of $n$ photons at a time. This phenomenon is a generalization of the well-known phenomenon of squeezing, in which photons are generated in pairs~\cite{Yurke,Movshovich,Yamamoto,Vahlbruch,Kudra}. The theoretical description of generalized squeezing has proven to be a difficult problem, resulting in debates that have continued for decades~\cite{Fisher,Hillery1984,Hong,Braunstein1987,Braunstein1990,Hillery1990,DellAnno,Zelaya,Braak}.

One of the authors recently performed simulations of generalized squeezed states in finite-dimensional truncations of Fock spaces, and found oscillatory dynamics of the mean photon number for higher-order squeezing, in contrast to the unbounded photon generation in the case of two-photon squeezing~\cite{Ashhab2025}. A few other remarkable and counter-intuitive results came out of that study. For example, in the cases of three- and four-photon squeezing, if we extrapolate the simulation results to the limit of an infinite Hilbert space size, the oscillation amplitude diverges, while the oscillation period converges. In the five-photon case, as well as higher orders beyond five, both oscillation amplitudes and periods were instead well-behaved in the simulations. Specifically, these quantities exhibited no noticeable change when the simulation size was changed by orders of magnitude. Such a situation would typically indicate that the simulations reliably capture the physical behaviour of the system under study.
Moreover, the oscillations indicated a \emph{discrete} spectrum of the higher-order squeezing operators, in contrast to the continuous spectrum of ordinary squeezing.

Subsequent to the publication of Ref.~\cite{Ashhab2025}, the authors of Ref.~\cite{Gordillo} and we independently noticed a serious issue with the simulation of higher-order squeezing. In all the simulations of Ref.~\cite{Ashhab2025}, the truncation dimension was an even number. If we instead perform simulations with an odd size, the results converge as well for large dimensions, but to a different value. One normally expects the parity of the dimension (i.e.~odd vs even) to have a negligible effect on the result for sufficiently large simulation sizes. Therefore, the strong parity dependence of the simulation results raises questions about the reliability and physical correctness of these numerical simulations and, more generally, inspires a careful mathematical analysis of the generalized squeezing problem.

In this paper, we present our recent results on the simulation of generalized squeezing with different truncation sizes, and we discuss the implications of our new findings in relation to the physical interpretation of the results in light of the mathematical theory of Hermitian operators in infinite-dimensional Hilbert spaces. We also present simulation results pertaining to the spectrum of the squeezing operator that help explain the different dynamics of the even/odd case. Furthermore, we present a mathematical derivation that identifies the even- and odd-truncation-size results as the difference between two different self-adjoint extensions of the squeezing Hamiltonian, which itself is not self-adjoint on states with finitely many excitations. Moreover, we show that the addition of Kerr terms restores the essential self-adjointness of the Hamiltonian, and thus a unique and parity-independent dynamics.

The remainder of this paper is organized as follows. In Section~\ref{Sec:Theory} we describe the model of the generalized squeezing operator. In Section~\ref{Sec:Simulations} we present our simulation results and discuss a few related questions, such as the physical meaning of the finite and infinite Hilbert space dynamics, relation to recent experiments, and the effect of adding a regulator term to make the extrapolation to infinite dimensions unique. In Section~\ref{Sec:Math} we provide a mathematical analysis that explains both the parity dependence of the simulation results in the absence of additional regulator terms, and the regularizing effect of such additional terms on the dynamics. We conclude with some final remarks in Section~\ref{Sec:Conclusion}.

\section{Theoretical background}
\label{Sec:Theory}

As explained for example in Ref.~\cite{Ashhab2025}, the $n$th-order generalization of squeezing, i.e.~the general case in which photons are created in groups of $n$ photons, is described by the squeezing operator of order $n$
\begin{equation}
\hat{U}_n \left( r \right) = \exp \left\{ -i r \hat{H}_n \right\},
\label{Eq:GeneralizedSqueezingOp}
\end{equation}
where
\begin{equation}
\hat{H}_n = i \left[ \left(\hat{a}^\dagger\right)^n - \hat{a}^n \right],
\label{Eq:Hamiltonian_n}
\end{equation}
$r$ is the squeezing parameter, and $\hat{a}$ and $\hat{a}^{\dagger}$ are, respectively, the photon annihilation and creation operators of a single mode of the electromagnetic field. The operator (\ref{Eq:Hamiltonian_n}) is Hermitian but unbounded and acts on an infinite-dimensional Hilbert space, which is spanned by eigenstates $\ket{m}$, $m=0,1,2\ldots\infty$, of the photon number operator $\hat{a}^\dagger \hat{a}$ (Fock states). We refer to Section~\ref{Sec:Math} for additional mathematical details.

In our simulations in this work, we focus on the case of squeezed vacuum states $\ket{\psi_n(r)}$, in which the squeezing operator $\hat{U}_n(r)$ is applied to the vacuum state $\ket{0}$, i.e.~$\ket{\psi_n(r)}=\hat{U}_n(r)\ket{0}$. The operator (\ref{Eq:GeneralizedSqueezingOp}) will be unitary for real $r$, according to Stone's theorem, if its generator $\hat{H}_n$ is self-adjoint~\cite{reed1}. While all Hermitian operators in finite-dimensional state spaces are self-adjoint, this is not necessarily the case for unbounded operators in infinite dimensions on a given domain~\cite{reed1}. It may happen that a Hermitian operator $\hat{A}$ is not self-adjoint, and therefore does not generate a unique unitary transformation $\exp(-ir\hat{A})$. However, all numerical simulations inherently take place in finite-dimensional spaces, where this difference does not arise, and it remains necessary to check whether an extrapolation to infinite dimension is possible. A necessary condition for this is the convergence of the computed results as the state space dimension increases. In most familiar cases, the operator $\hat{A}$ is well-behaved with increasing dimension, and the corresponding evolution operator is unique, leading to unique and converging extrapolations of finite-dimensional simulations. We will see in the following that the generalized squeezing operator $\hat{H}_n$ for $n\ge 3$ does not have this property. 

The simulations of the dynamics proceed as in Ref.~\cite{Ashhab2025}, by constructing a truncated Hamiltonian matrix of size $N \times N$. We note here that the definition of $N$ in this work is different from that of Ref.~\cite{Ashhab2025}. There, a simulation size $N$ meant that the number of photons ranged from $0$ to $N-1$. In the present work, since we always set the initial state to the zero-photon state, and photons are created and annihilated in groups of $n$ photons, we focus instead on the Hilbert space $\left\{ \ket{0}, \ket{n}, \ket{2n}, \cdots, \ket{n\times (N-1)} \right\}$, which contains $N$ basis states. This separation of the state space into $n$ invariant subspaces persists in the infinite-dimensional Hilbert space, so that we can safely restrict the calculations to the subspace containing $\ket{0}$. Hence, the truncated Hamiltonian reads $\hat{H}_n^{(N)}=\hat{P}_N\hat{H}_n\hat{P}_N$, where $\hat{P}_N$ is the $N$-dimensional projector onto the first invariant subspace containing $\ket{0}$. We then evaluate the squeezing operator with squeezing parameter having small value $\delta r=0.01$, which gives $\hat{U}_n^{(N)} (\delta r)=\exp\left\{-i\delta r \hat{H}^{(N)}_n\right\}$. Using this operator, we evaluate $\hat{U}^{(N)}_n(r)$ for any value of $r$ using the formula $\hat{U}^{(N)}_n (r)=\left[\hat{U}^{(N)}_n (\delta r)\right]^{r/\delta r}$. For this formula to work straightforwardly, we use $r$ values that are integer multiples of $\delta r$. Following the above procedure, we obtain the squeezed vacuum state $\ket{\psi_n^{(N)}(r)}$ for a range of $r$ values. As a representative physical quantity, in this work we focus on the average photon number $\bra{\psi_n^{(N)}(r)} \hat{a}^{\dagger} \hat{a} \ket{\psi_n^{(N)}(r)}$, which we plot as a function of $r$ to analyse the dynamics.

\section{Simulation results}
\label{Sec:Simulations}

In this section, we present simulation results that demonstrate the parity dependence of the squeezing dynamics, i.e.~the unitary operator $\hat{U}^{(N)}_n$. We also present the results of additional simulations that probe the spectrum of the generator $\hat{H}_n^{(N)}$.

\subsection{Parity dependence}
\label{Sec:SimulationsParity}

\begin{figure}[h]
\includegraphics[width=8cm]{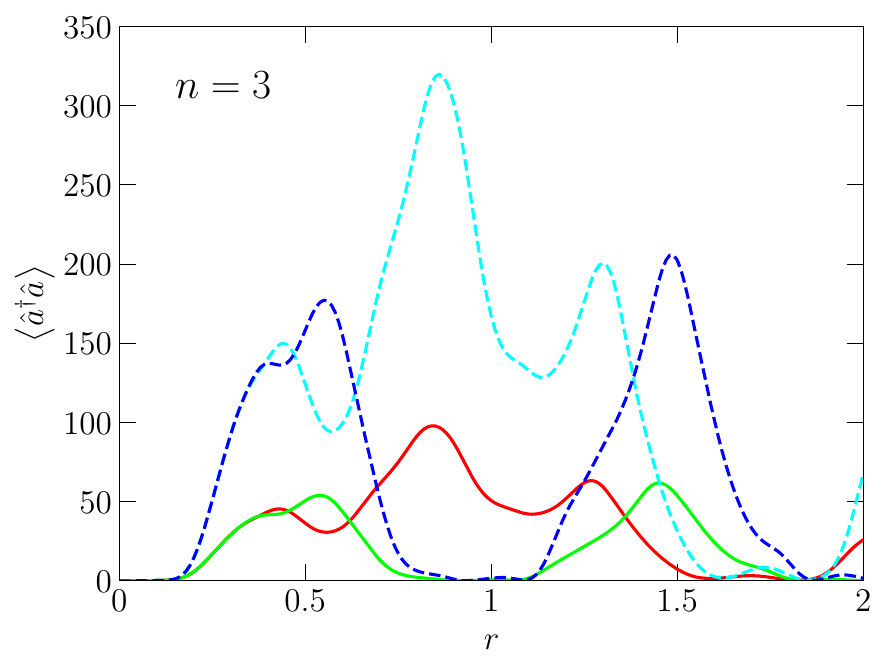}
\includegraphics[width=8cm]{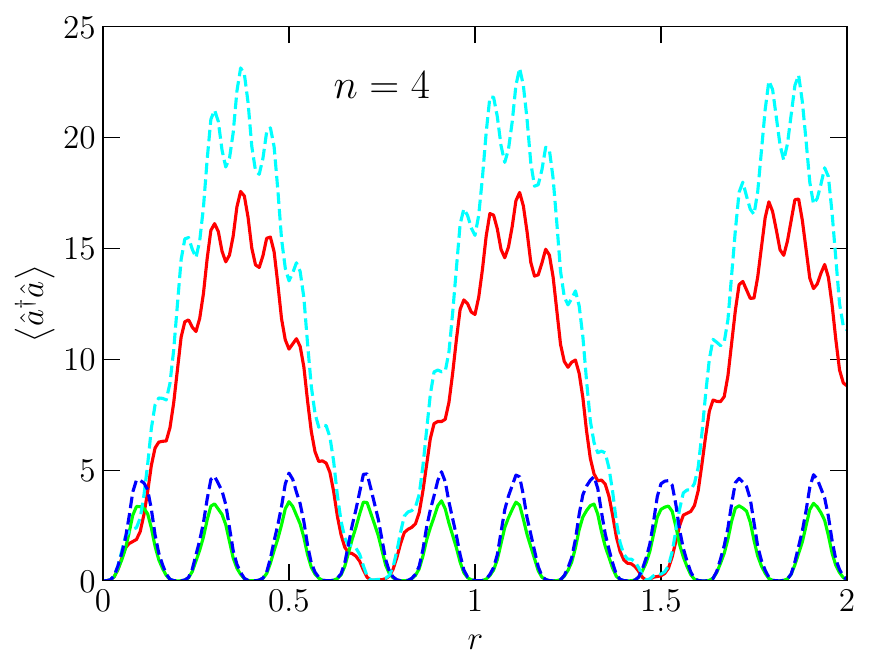}
\includegraphics[width=8cm]{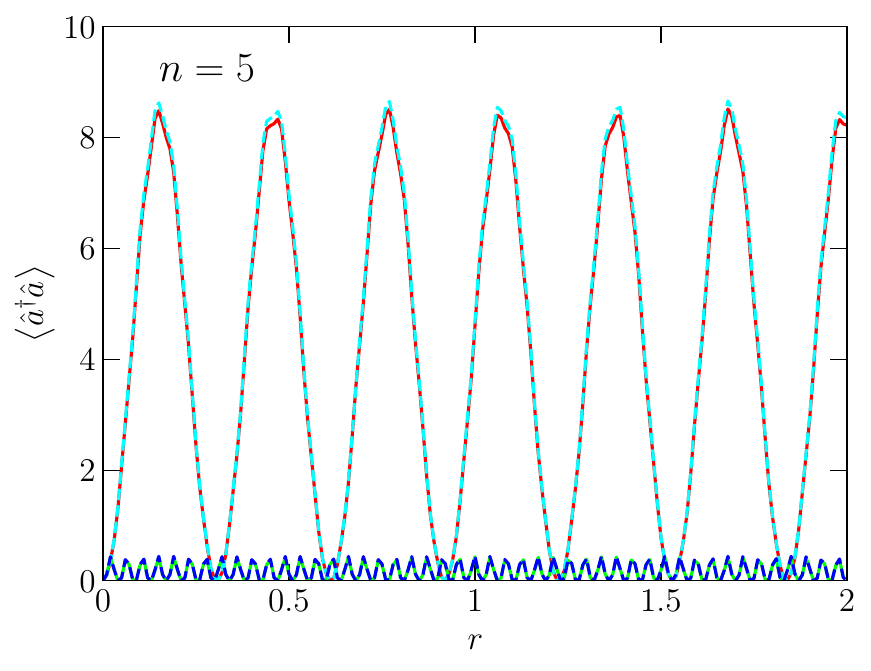}
\includegraphics[width=8cm]{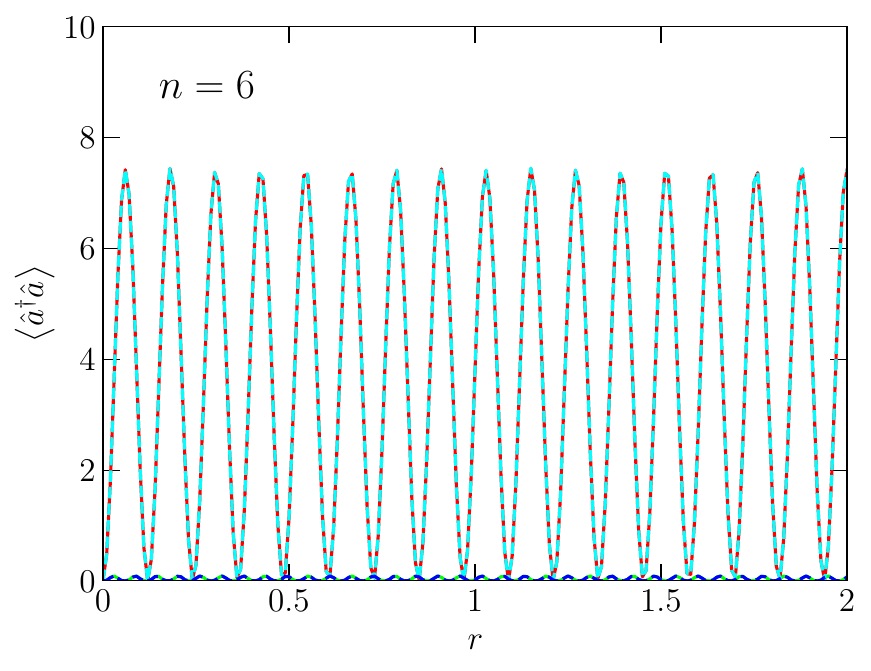}
\caption{Average photon number $\left\langle \hat{a}^{\dagger} \hat{a} \right\rangle$ for the state $\hat{U}^{(N)}_n(r)\ket{0}$ as a function of the squeezing parameter $r$. The red, green, cyan and blue lines correspond, respectively, to $N=1000$, 1001, $10^4$ and $10^4+1$. Although not immediately visible in the figure, the $N=1001$ and $10^4+1$ simulation results in the cases $n=5,6$, as well as the $N=1000$ and $10^4$ simulation results, agree with each other. The maximum height of the odd-$N$ oscillations is 0.4 and 0.09 for $n=5$ and 6, respectively.}
\label{Fig:AveragePhotonNumberVsSqueezingParameter}
\end{figure}

We start by plotting the average photon number $\bra{\psi_n^{(N)}(r)} \hat{a}^{\dagger} \hat{a} \ket{\psi_n^{(N)}(r)}$ as a function of the squeezing parameter $r$ in Fig.~\ref{Fig:AveragePhotonNumberVsSqueezingParameter}. If we focus on the red and cyan lines, we have the situation discussed in Ref.~\cite{Ashhab2025}. For $n=3$ and $n=4$, the amplitude of the oscillations depends on the truncation size $N$ in the simulation. Furthermore, the amplitude seems to diverge when the results are extrapolated to $N\rightarrow\infty$. However, the overall shape of the oscillations, including the period, are independent of $N$ for sufficiently large values of $N$. This dichotomy in the behaviours of the oscillation amplitude and period is quite unusual. For larger values of $n$ ($n\geq 5$), both the amplitude and the period appear to converge well. The $n=5$ plot in Fig.~\ref{Fig:AveragePhotonNumberVsSqueezingParameter} shows a very small change in amplitude when we go from $N=10^3$ to $N=10^4$. The $n=6$ plot shows no discernible change in the transition from $N=10^3$ to $N=10^4$. One would deduce from these results that, at least for $n\ge5$, the limit $N\rightarrow\infty$ is well-defined for $\hat{U}_n^{(N)}(r)$ and can be reliably obtained from simulations in finite-dimensional spaces.

A remarkable change occurs when we consider odd values of the truncated Hamiltonian size $N$. All the plots in Fig.~\ref{Fig:AveragePhotonNumberVsSqueezingParameter} exhibit a dramatic change in the transition from even to odd values of $N$. Odd values of $N$ lead to faster, lower-amplitude oscillations. This tendency is especially clear for $n\geq 4$. Interestingly, if we focus on odd-$N$ simulations, the simulations with different values of $N$ behave similarly to what we observed with even values of $N$. For $n=3$ and $n=4$, the oscillation amplitude depends on $N$, while this dependence becomes weaker and eventually disappears as we increase $n$. For $n\ge 5$ we observe again a convergent behaviour of the squeezing dynamics which becomes independent of $N$ for $N\rightarrow\infty$. One would again conclude that $\hat{U}_n^{(\infty)}(r)$ is well defined -- but obviously this operator is \emph{different} from the $\hat{U}_n^{(\infty)}(r)$ obtained with even $N$.

We are forced to conclude that the extrapolation of the finite-dimensional operator $\hat{U}_n^{(N)}(r)$ to $N\rightarrow\infty$ is \emph{not unique}: there are two different operators in this limit, respectively obtained with even and odd $N$. Which one of them corresponds to the physical (untruncated) Hilbert space? One could argue that the assumption of an unbounded photon number in any experimental setup is itself unphysical, and the ``real" operator corresponding to a given implementation of $\hat{H}_n$ is just one of the $\hat{H}_n^{(N)}$. However, it is obvious that a hard cut-off of the photon number which moreover fixes its parity does not correspond to any realistic experiment. Starting from a harmonic oscillator driven by the appropriate operator at the appropriate frequency to realize the effective Hamiltonian in Eq.~\eqref{Eq:Hamiltonian_n}, one must deal with a potentially unbounded photon number, and our results show that the squeezed state for $n\ge3$ is not uniquely obtainable from finite-size simulations. This means that the Hamiltonian (\ref{Eq:Hamiltonian_n}) has no physical interpretation without further information about the actually realized system, which cannot be described solely by the expression of $\hat{H}_n$. 

For example, when modelling superconducting circuits based on Josephson junctions, we typically set up a mathematical description which employs phase (or flux) and charge variables of the electric circuit. This model can be mapped to the problem of a single particle moving in a complicated potential energy landscape in a high-dimensional space. Then, we derive an approximate Hamiltonian that describes the one or few degrees of freedom that are relevant to the physical behaviour of the circuit. The derivation typically assumes that the system remains sufficiently close to its ground state with only a small number of excitations. To be more precise, it is in fact possible to design superconducting harmonic oscillators (not containing Josephson junctions) that can be populated by very large numbers of photons. The difficulty arises when we add Josephson junctions to introduce nonlinearity into the system. Then we obtain a complicated Hamiltonian with many nonlinear terms. If we assume that one of the experimentally controllable parameters is modulated at the appropriate frequency for multi-photon generation, we can make the rotating-wave approximation to obtain Eq.~(\ref{Eq:Hamiltonian_n}). Importantly, multiple steps in this derivation rely on the assumption that the number of photons remains small. As a result, the model Hamiltonian in Eq.~(\ref{Eq:Hamiltonian_n}) cannot remain valid up to infinite photon numbers. If we excite the circuit above a certain energy scale, the circuit will generally behave completely differently. Even before we reach the energy scale of complete breakdown of the model, as the photon number increases, other nonlinear terms in the Hamiltonian that are ignored in the derivation of Eq.~(\ref{Eq:Hamiltonian_n}) will affect the dynamics and lead to modified dynamics that depend on these additional terms. We will provide a few specific examples of this situation below. 

To put this discussion in context, the authors of Ref.~\cite{Eriksson} estimated that the idealized Hamiltonian in their experiment was reliably valid for a Hilbert space extending up to a few photons only. Our simulations go orders of magnitude beyond this realistic scale. It is therefore important to keep in mind that the theoretical model of Eq.~(\ref{Eq:Hamiltonian_n}) extending to extremely large photon numbers cannot be realized physically.

It is worth noting here that the operator in Eq.~(\ref{Eq:GeneralizedSqueezingOp}) with $n=1$ and $n=2$ has been extensively used in the study of quantum optics, even when dealing with systems that are expected to exhibit nonlinear behavior at large photon numbers, such as superconducting circuits. An important difference in the cases $n=1$ and $n=2$ is that, for sufficiently weak nonlinearities, the predictions of the model become independent of the Hamiltonian details at large photon numbers. Therefore, simulation results for short evolution times become independent of the truncation details if the truncation size is sufficiently large and the nonlinearity is sufficiently weak. In this case, one can avoid dealing with the questions of truncation size and nonlinearity in a simplified model, whereas these questions are unavoidable for $n\geq 3$. The mathematical reason for the independence of truncation details for $n=1$ and $n=2$, as opposed to the cases $n\geq3$, is the self-adjointness of $\hat{H}_1$ and $\hat{H}_2$, which will be discussed in Section \ref{Sec:Math}. 

One might wonder if the mathematical complications encountered above can be avoided in a natural way by using the real-space representation, i.e.~solving the Schr\"odinger equation
\begin{equation}
i \frac{\partial\psi(x,t)}{dt} = \frac{i}{2^{n/2}} \left[ \left(x - \frac{\partial}{\partial x} \right)^n - \left(x + \frac{\partial}{\partial x} \right)^n \right] \psi(x,t).
\label{Eq:SchroedingerEquationRealSpace}
\end{equation}
The squeezing parameter $r$ would be proportional to the total evolution time in this description. Since the (untruncated) problems in the real-space and Fock-space representations are isomorphic to each other, the non-uniqueness of the operator $\hat{U}_n(r)$ must also appear in the real-space representation. Another important point to note in this context is that, although the ladder operator description of the harmonic oscillator is sometimes derived from an underlying real-space description, we cannot say that the real-space representation should take precedence over the Fock-space representation. For example, Eq.~(\ref{Eq:Hamiltonian_n}) is not derived as a rewriting of a real-space Hamiltonian of the form in Eq.~(\ref{Eq:SchroedingerEquationRealSpace}), but rather by working in the Fock-space representation, and making a few assumptions and approximations within the derivation of Eq.~(\ref{Eq:Hamiltonian_n}), notably the rotating-wave approximation. From this point of view, the Fock-space representation is in fact a more natural framework for discussing the squeezing problem.

\begin{figure}[h]
\includegraphics[width=8cm]{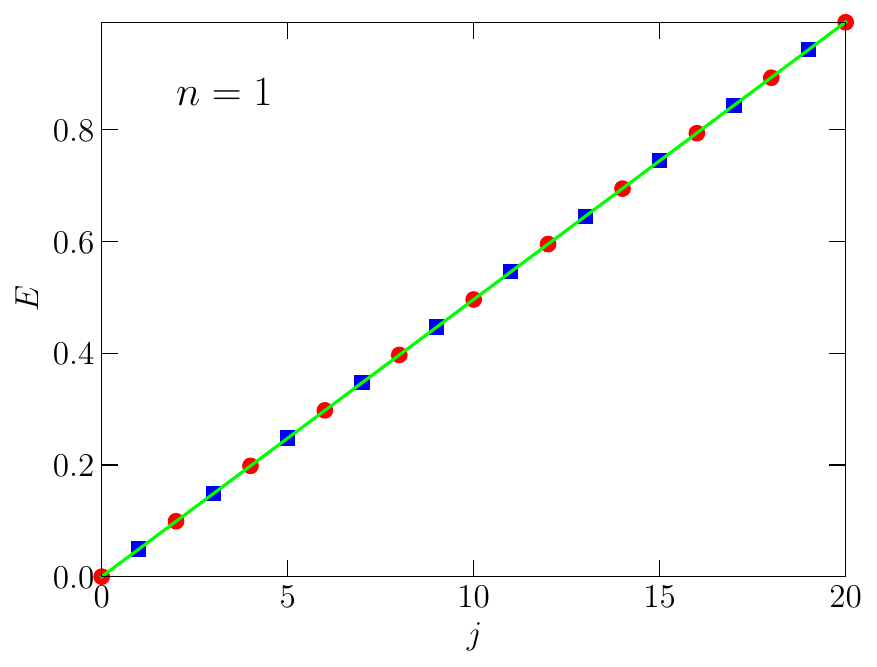}
\includegraphics[width=8cm]{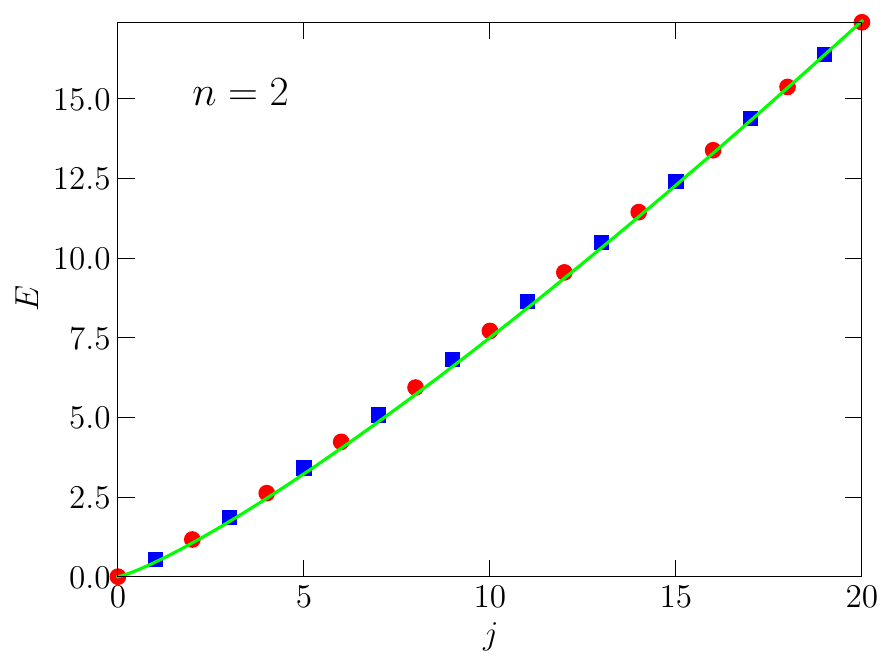}
\includegraphics[width=8cm]{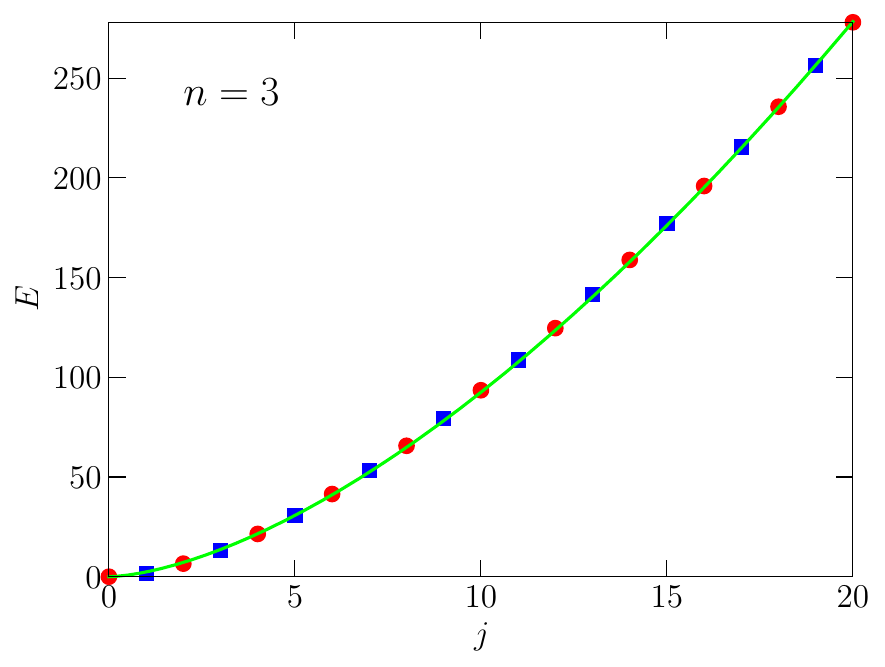}
\includegraphics[width=8cm]{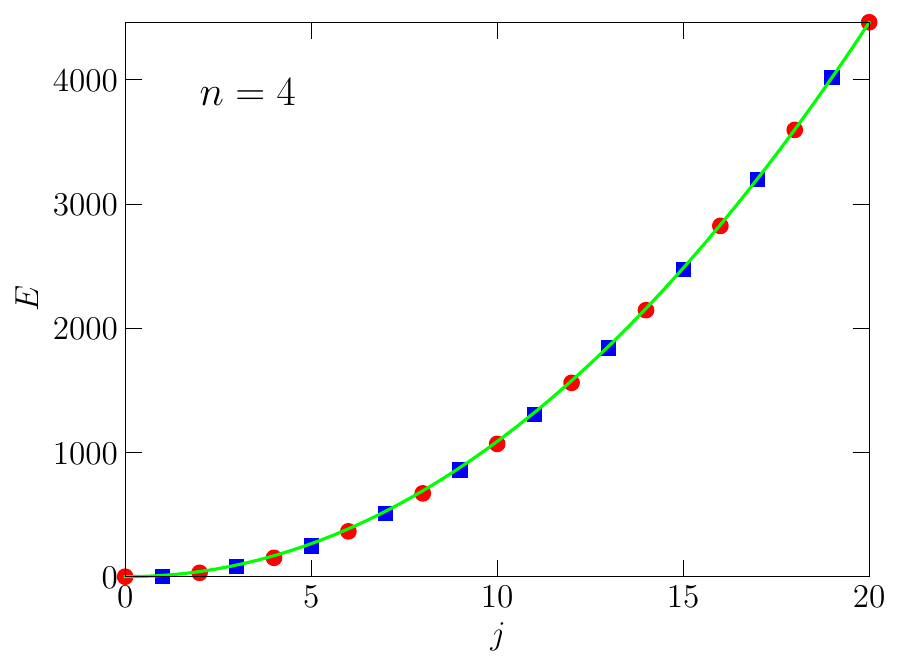}
\caption{Eigenvalues of $\hat{H}_n$ in the middle of the spectrum for $n=1$, 2, 3 and 4. Here we plot the few lowest non-negative eigenvalues in each case. In each panel, we plot the eigenvalues for two truncation sizes together. The two truncation sizes are $N=1000$ (blue squares) and 1001 (red circles). The green lines are fits of the form $E=\alpha j^\gamma$ with fitting parameters $\gamma=1.001$, 1.217, 1.590 and 2.035 for $n=1$, 2, 3 and 4, respectively.}
\label{Fig:SpectrumOddEvenCombined}
\end{figure}

We now move on to analyse the parity dependence of the simulations from a different point of view. Clearly, the observed dynamics is closely related to the spectrum of the Hamiltonian $\hat{H}^{(N)}_n$. In Fig.~2 we plot the eigenvalues near the middle of the spectrum, i.e.~near zero, for neighbouring even and odd values of $N$. These spectra reveal one feature: the spectrum is symmetric about zero, so that every positive eigenvalue has a negative counterpart with the same absolute value. This symmetry of the spectrum can be proven both for $\hat{H}_n$ and $\hat{H}_n^{(N)}$. As a result, the odd-$N$ case must have at least one zero eigenvalue. Since there is no reason for the spectrum to have degeneracies, which is confirmed by our numerical calculations, it is natural that the even- and odd-$N$ cases have different spectra. This result, in turn, makes it less surprising that the two cases produce different dynamics for finite $N$. 

The analysis of these spectra reveals a few additional interesting features. By observing each spectrum as a whole, we find that, focusing on the non-negative eigenvalues and excluding the eigenvalues near the upper end of the spectrum, the eigenvalues follow a power-law function of the form $E_j=\alpha j^{\gamma}$. The exponent in the fitting function is the same for both odd and even values of $N$. Furthermore, if we take eigenvalues in an alternating order from one odd- and one even-$N$ spectrum, and we plot the resulting list of eigenvalues, these are accurately described by a fitting function with the same power-law exponent. Currently, we have no mathematical explanation for this surprising fact. For $n=1,2$, our observation is only relevant for finite $N$, because the operators $\hat{H}_{1,2}$ have real and continuous spectra, extending from $-\infty$ to $+\infty$.

One more point worth noting here in relation to the spectrum is the fact that the eigenvector associated with the zero eigenvalue in the odd-$N$ case has a large weight at zero photon number, a tendency that becomes increasingly strong with increasing $n$. As such, the vacuum state $\ket{0}$ has a large overlap with the zero-eigenvalue state. This fact explains why the oscillation amplitude is small in the odd-$N$ case, especially as we increase $n$. Similarly, for the even-$N$ case, the vacuum state $\ket{0}$ is almost entirely a superposition of the two eigenvectors at the center of the spectrum, and these two eigenvectors are almost entirely localized in the lowest two Fock states. Another interesting observation from our numerical simulations is that for $n\geq 3$ and even $N$, the smallest positive eigenvalue approaches $\sqrt{n!}$ with increasing $n$, which suggests that the lowest two Fock states decouple from the higher Fock states.

\subsection{Regulating the dynamics via Kerr interaction}
\label{Sec:Simulations_Kerr}

\begin{figure}[h]
\includegraphics[width=8cm]{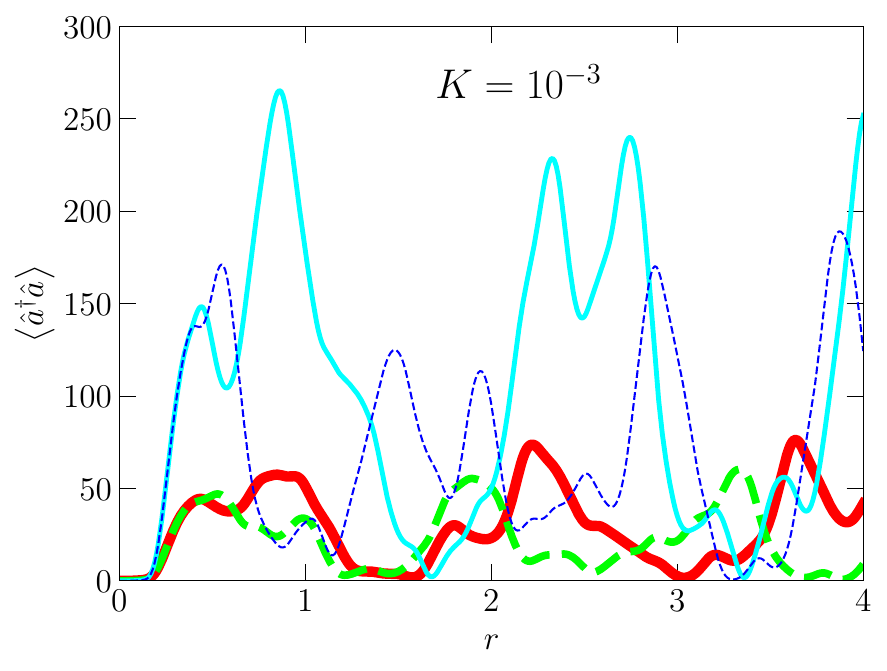}
\includegraphics[width=8cm]{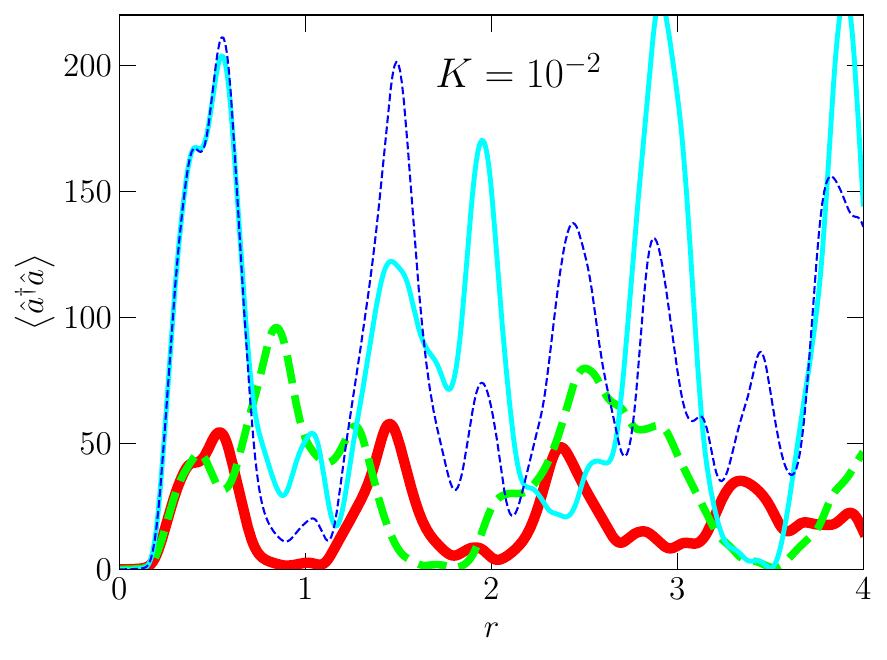}
\includegraphics[width=8cm]{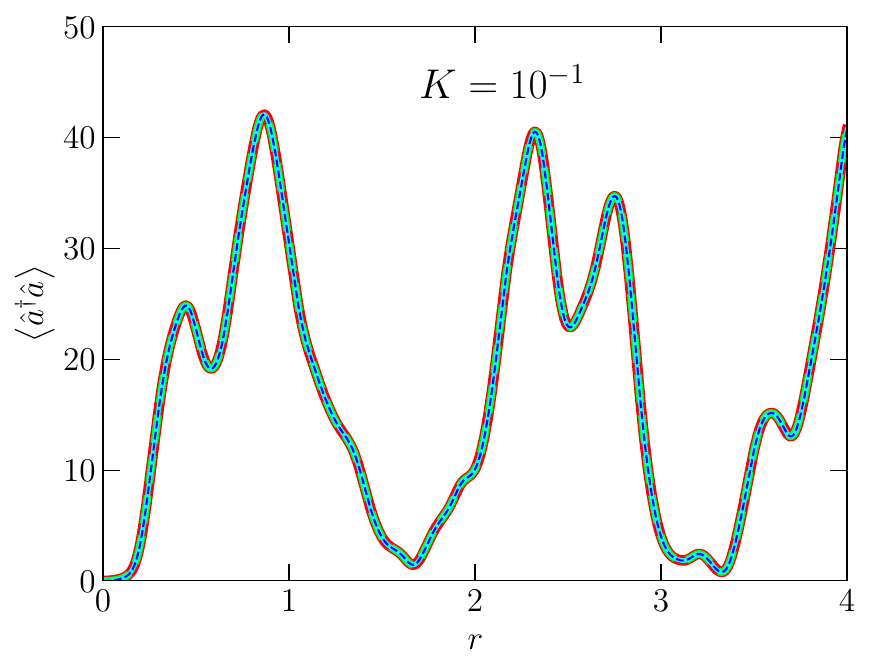}
\includegraphics[width=8cm]{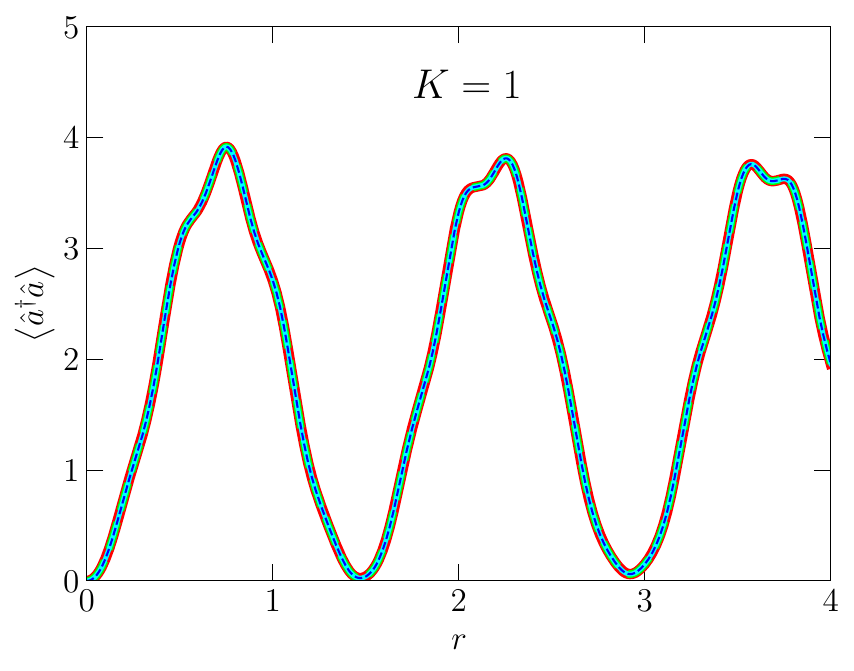}
\caption{Average photon number $\left\langle \hat{a}^{\dagger} \hat{a} \right\rangle$ for the state $\hat{U}_3^{(N)}(r)\ket{0}$ as a function of the squeezing parameter $r$ with a quadratic Kerr term of varying strength ($K$). The red, green, cyan and blue lines correspond, respectively, to $N=1000$, 1001, $10^4$ and $10^4+1$. The dynamics becomes independent of the truncation size when the largest Kerr (diagonal) matrix element in the Hamiltonian becomes larger than the largest squeezing (off-diagonal) matrix element.}
\label{Fig:AveragePhotonNumberVsSqueezingParameterKerr3Quadratic}
\end{figure}

The parity dependence of the results of Section \ref{Sec:SimulationsParity} unambiguously shows that $\hat{H}_n$ alone cannot describe the physically realized situation. In fact, as mentioned above, the actual Hamiltonian describing any realistic physical system will inevitably have additional terms that might be negligible for small photon numbers, but can no longer be ignored for sufficiently large photon numbers. With this point in mind, we investigate the possibility of using an alternative Hamiltonian generating a well-behaved evolution as we increase the truncation size to infinity. Specifically, we add a diagonal term to the Hamiltonian, and we choose a term that grows sufficiently fast that it creates a natural cutoff and ensures that the infinite-photon-number regime is not reached during the evolution. It is worth noting that the addition of this term is related to the stabilization of the $n$-photon quantum Rabi model recently studied in Refs.~\cite{Ying,Ayyash2025}.

As a representative example, we consider the tri-squeezing Hamiltonian with an added quadratic Kerr term:
\begin{equation}
\hat{H}_{3, \rm quadratic \ Kerr} = i \left[ \left(\hat{a}^\dagger\right)^3 - \hat{a}^3 \right] + K \left(\hat{a}^\dagger\right)^2 \hat{a}^2.
\label{Eq:Hamiltonian_3_QuadraticKerr}
\end{equation}
where the Kerr coefficient $K$ is positive. In principle, regardless of how small $K$ is, since the Kerr term has more creation and annihilation operators than the squeezing term (four vs three), the Kerr term will be dominant in the infinite-photon-number limit. As a result, it will create a physical energy barrier that will prevent the photon number from reaching infinity. We will see in Section \ref{Sec:Math_Kerr} that the Hamiltonian in Eq.~(\ref{Eq:Hamiltonian_3_QuadraticKerr}) is essentially self-adjoint on Fock states, and therefore its finite-dimensional truncations generate an evolution with a unique limit $N\rightarrow\infty$. The results of our simulations with this Hamiltonian are shown in Fig.~\ref{Fig:AveragePhotonNumberVsSqueezingParameterKerr3Quadratic}. When $K\sim 10^{-2}$, the dynamics is still strongly dependent on $N$ for $N$ up to $10^4$. However, when we increase $K$, we find that all the different simulations with different truncation sizes eventually produce the same results. This indicates that the Kerr term is producing the expected effect: it creates a natural cutoff making the simulation results insensitive to the exact details of the Hamiltonian at large photon numbers. The value of $K$ at which this cutoff effect occurs can be estimated as follows: at the upper photon number end of the Hamiltonian, we replace each creation or annihilation operator by $\sqrt{nN}$. The squeezing term is then of order $(nN)^{3/2}$, while the Kerr term is of order of $K(nN)^2$. If we want the Kerr term to dominate over the squeezing term, we must have $K(nN)^{1/2}>1$. The results in Fig.~\ref{Fig:AveragePhotonNumberVsSqueezingParameterKerr3Quadratic} do indeed obey this rule: when $K=10^{-2}$ and $N=1000$, we have $K(nN)^{1/2}=0.55$, so that the squeezing term is dominant, while when $K=10^{-1}$ and $N=1000$, we have $K(nN)^{1/2}=5.5$, so that the Kerr term is dominant. If we increase the Kerr coefficient further, the amplitude of the oscillations decreases, which is yet another indication that the Kerr term is acting as a natural cutoff that becomes increasingly confining with increasing $K$. Interestingly, the Kerr term not only removes the even-odd difference, but also the overall dependence of the oscillation amplitude on $N$, as can be seen in  Fig.~\ref{Fig:AveragePhotonNumberVsSqueezingParameter} for $K=0$. We note here that, even when $K$ is so large that the results are independent of $N$, the oscillations in Fig.~\ref{Fig:AveragePhotonNumberVsSqueezingParameterKerr3Quadratic} are in general irregular and strongly dependent on $K$, and only become quite regular when $K$ is so large that only a few quantum states are involved in the dynamics. We also note that the results for $N=10^4$ and $10^4+1$ begin to agree with each other at smaller values of $N$, as would be expected from the formula $K(nN)^{1/2}>1$. An interesting case corresponds to $K=10^{-2}$ in Fig.~\ref{Fig:AveragePhotonNumberVsSqueezingParameterKerr3Quadratic}, where the results for $N=10^4$ and $10^4+1$ agree only up to around $r=0.7$ and deviate from each other after that. This result serves as a reminder that small differences between different simulations will eventually lead to significantly different results for sufficiently large values of $r$.

Additional examples that further explore the role of nonlinear Kerr interactions in regulating the generalized squeezing dynamics are presented in Appendix A.

\subsection{Spectrum of the squeezing Hamiltonian}

\begin{figure}[h]
\includegraphics[width=8cm]{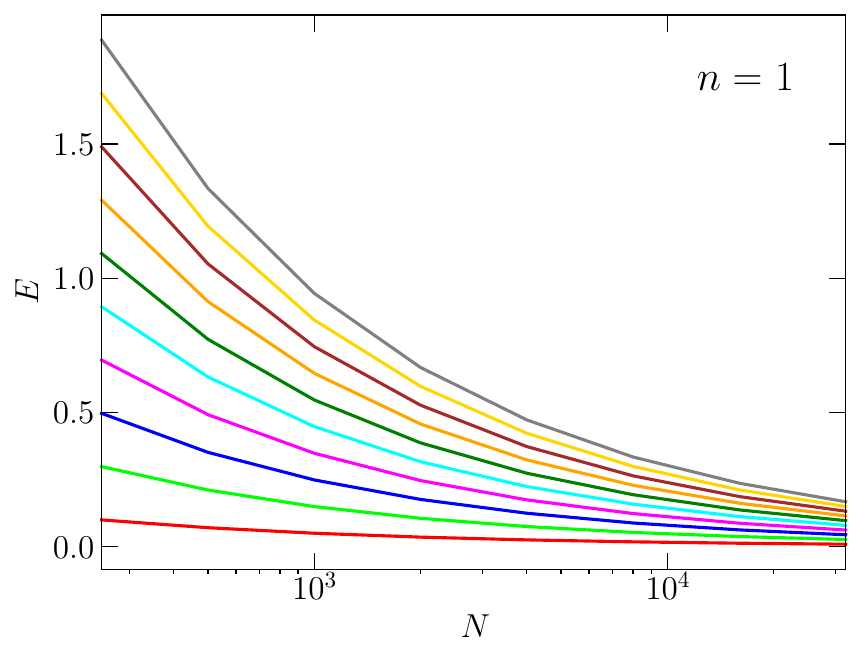}
\includegraphics[width=8cm]{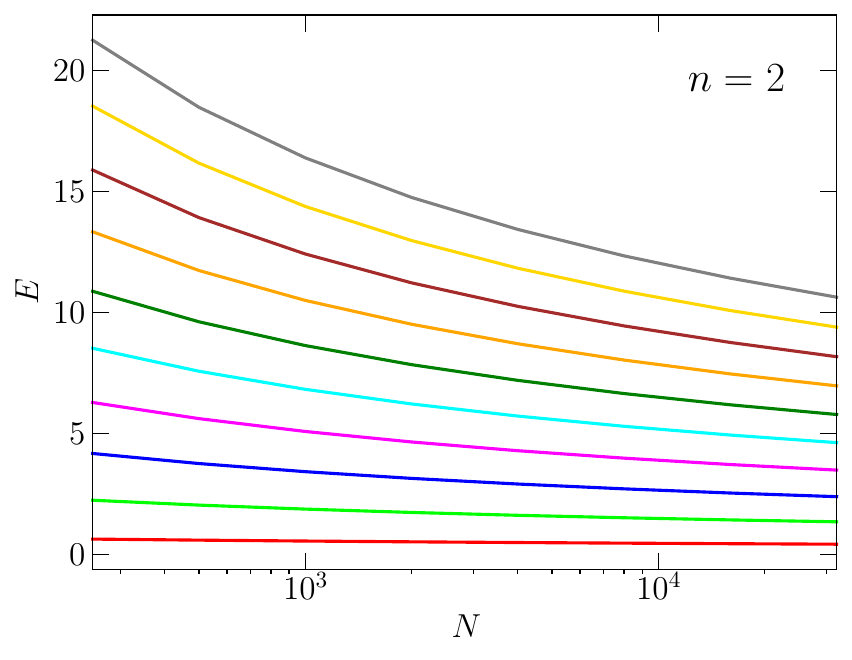}
\includegraphics[width=8cm]{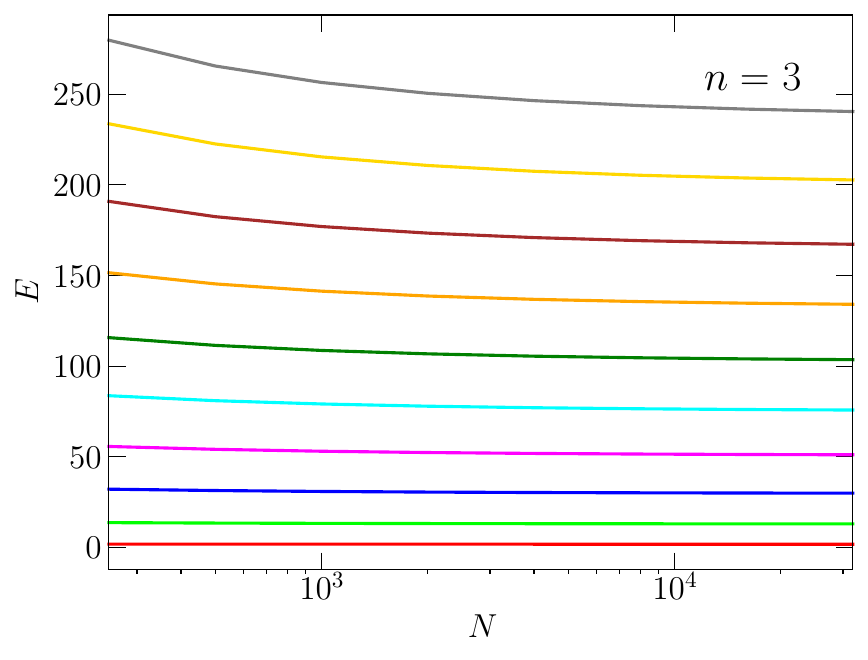}
\includegraphics[width=8cm]{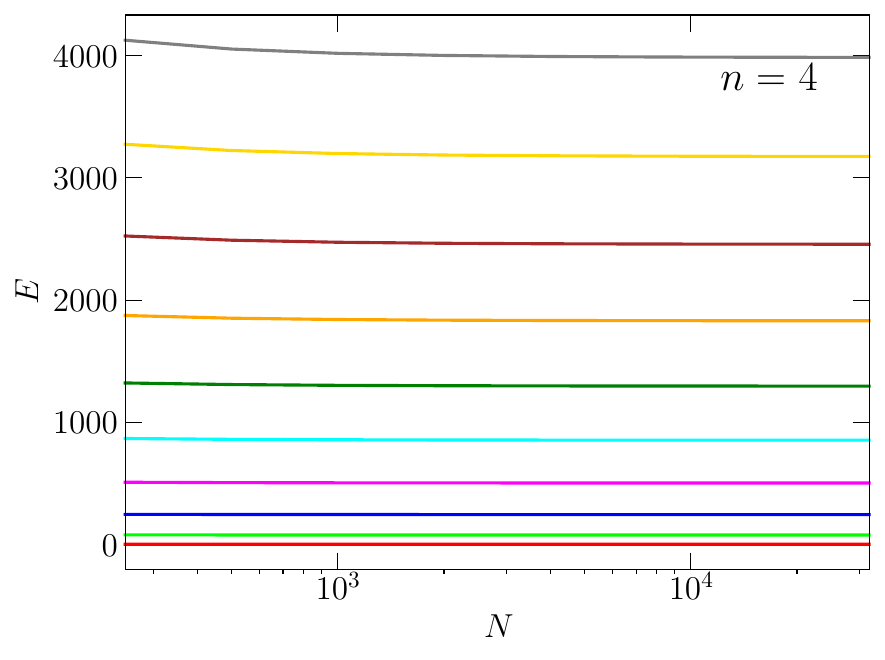}
\caption{Ten smallest positive eigenvalues of $\hat{H}_n$ as functions of truncation size $N$ for $n=1$, 2, 3 and 4. For all the data points, we chose even values of $N$.}
\label{Fig:SpectrumSmallestTenEigenvaluesVsTruncationSize}
\end{figure}

\begin{figure}[h]
\includegraphics[width=8cm]{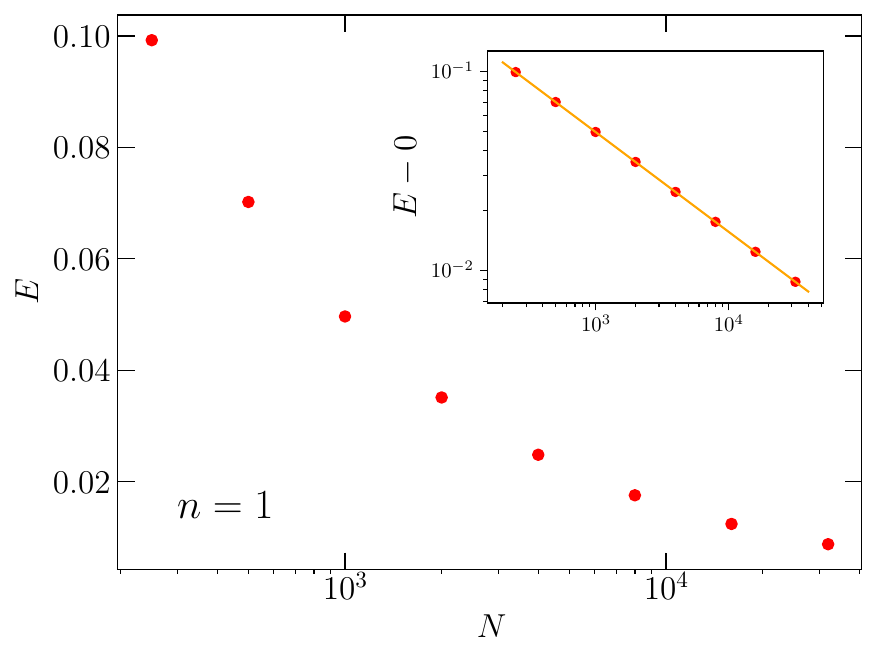}
\includegraphics[width=8cm]{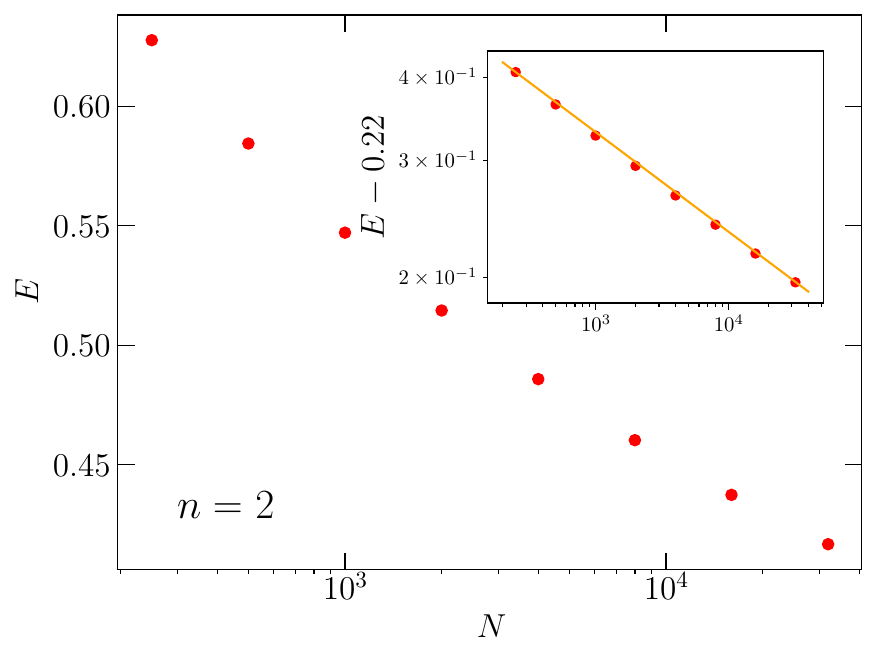}
\includegraphics[width=8cm]{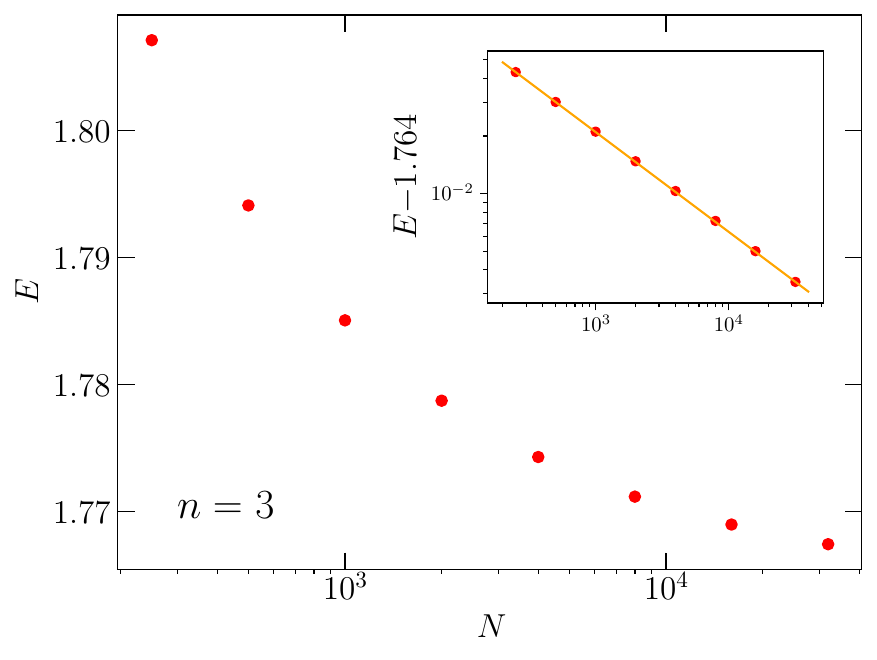}
\includegraphics[width=8cm]{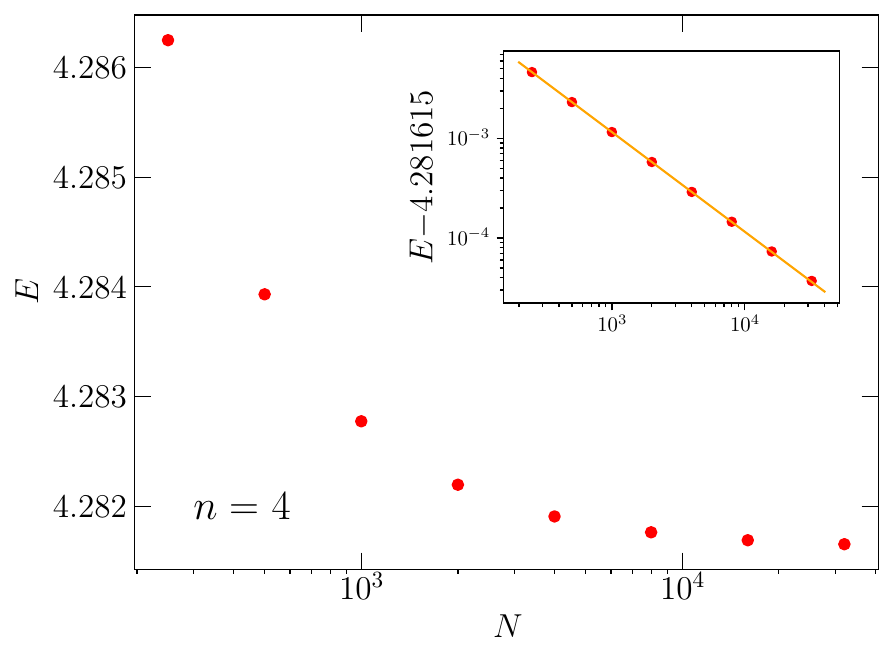}
\caption{Smallest positive eigenvalue of $\hat{H}_n$ as a function of truncation size $N$ for $n=1$, 2, 3 and 4. As in Fig.~\ref{Fig:SpectrumSmallestTenEigenvaluesVsTruncationSize}, we use only even values of $N$ in this figure. The insets show the same data in a log-log plot that demonstrates the asymptotic value of the data if extrapolated to $N\rightarrow\infty$. The orange lines are straight-line fits showing the quality of the agreement.}
\label{Fig:SpectrumSmallestEigenvalueVsTruncationSize}
\end{figure}

\begin{figure}[h]
\includegraphics[width=8cm]{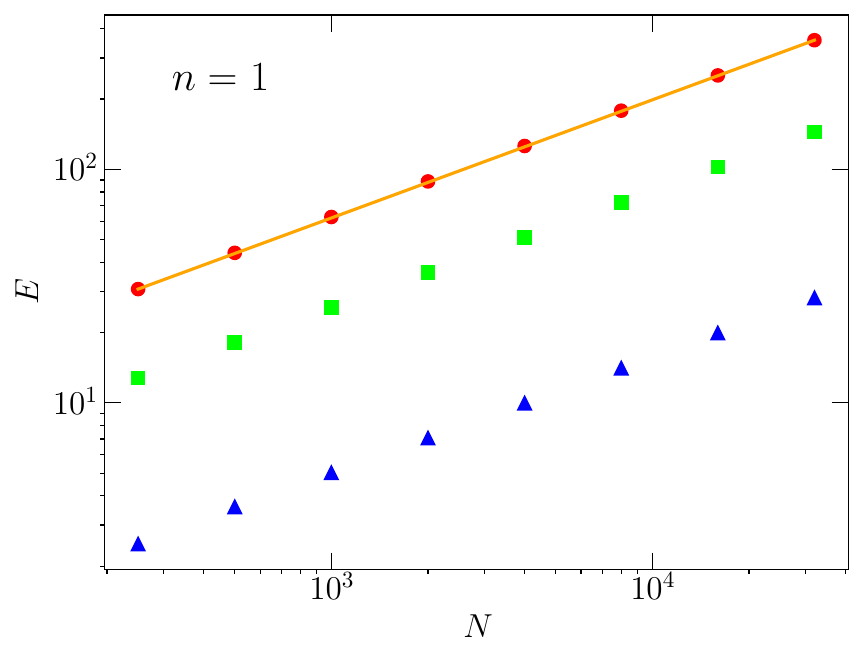}
\includegraphics[width=8cm]{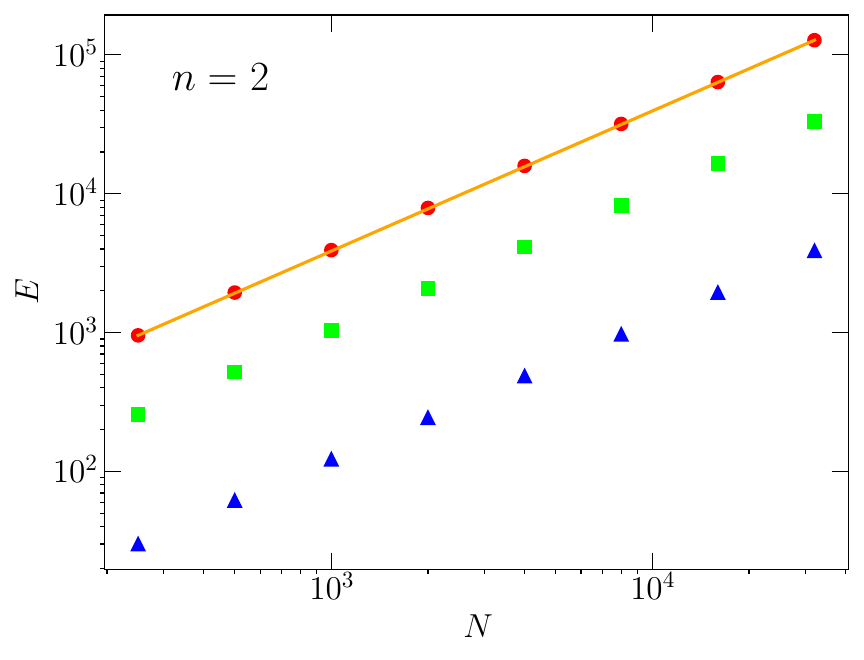}
\includegraphics[width=8cm]{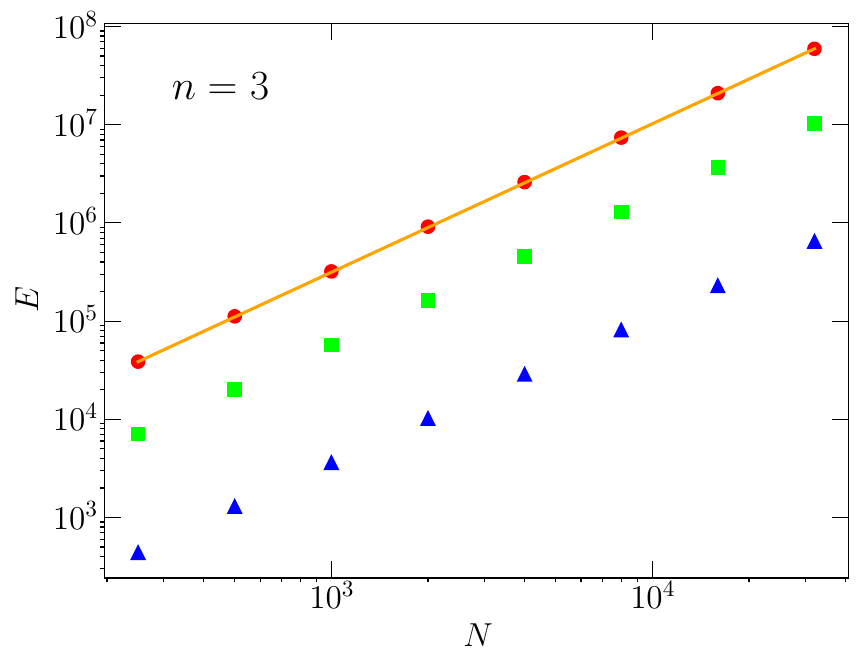}
\includegraphics[width=8cm]{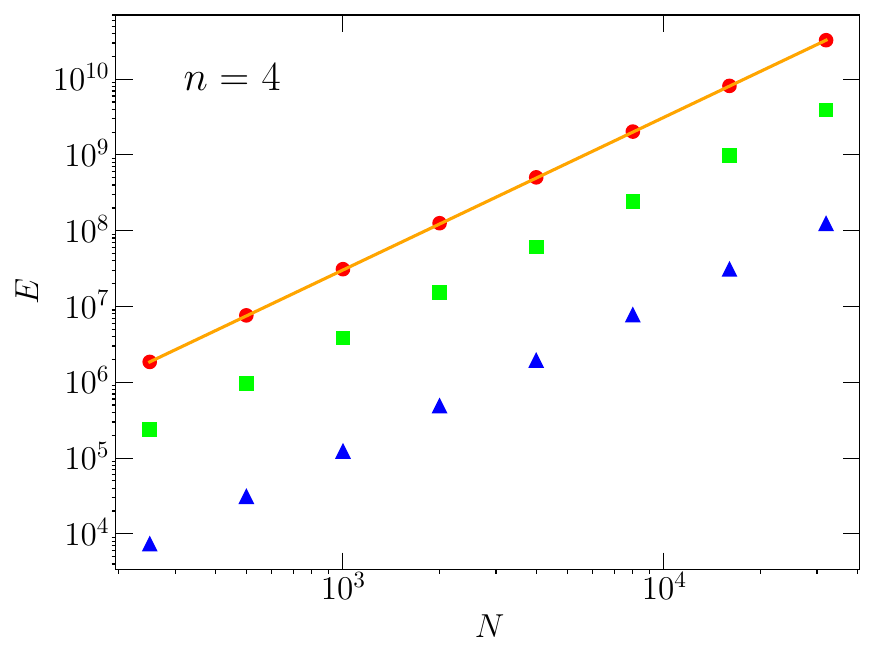}
\caption{Three large eigenvalues of $\hat{H}_n$ as functions of truncation size $N$ for $n=1$, 2, 3 and 4. The red circles, green squares, and blue triangles respectively correspond to the eigenvalues with index $N-1$ (largest eigenvalue), $3N/4$, and $11N/20$. (Note that the centre of the spectrum is at $N/2$.) The orange straight line corresponds to a power-law fit to the red circles. The slopes (i.e.~power-law exponents) are 0.51, 1.01, 1.51 and 2.01 for $n=1$, 2, 3 and 4, respectively. The green-square and blue-triangle data have the same slopes as the red-circle data in all cases.}
\label{Fig:SpectrumLargeEigenvaluesVsTruncationSize}
\end{figure}

We now consider the Hamiltonian $\hat{H}_n^{(N)}$ itself with respect to its (necessarily discrete) spectrum.
 In this analysis, we focus on the properties near the centre of the spectrum, i.e.~excluding the eigenvalues near the edges of the spectrum. We perform some calculations to address this question for $\hat{H}^{(N)}_n$ with different values of $n$. It is worth noting at this point that the cases $n=1$ (displacement operator) and $n=2$ (two-photon squeezing) are well-known in the literature, and $\hat{H}^{(N)}_{1,2}$  accurately predicts the behaviour of $\hat{H}_{1,2}$ for $N\rightarrow\infty$. We will nevertheless include these cases in the following analysis for comparison.

In Fig.~\ref{Fig:SpectrumSmallestTenEigenvaluesVsTruncationSize} we plot the ten smallest positive eigenvalues of $\hat{H}^{(N)}_n$ as functions of $N$. We restrict this analysis to even values of $N$. We performed similar calculations for odd values of $N$, but we do not show the results here, since they are similar to those of the even-$N$ case, although the spectra differ as shown in Fig.~\ref{Fig:SpectrumOddEvenCombined}. In Fig.~\ref{Fig:SpectrumSmallestEigenvalueVsTruncationSize} we focus on the smallest positive eigenvalue and try to identify its asymptotic value in the limit $N\rightarrow\infty$. In Fig.~\ref{Fig:SpectrumLargeEigenvaluesVsTruncationSize} we plot three representative large eigenvalues as functions of $N$ to get an impression of the overall size of the spectrum. By combining the results plotted in Figs.~\ref{Fig:SpectrumOddEvenCombined},~\ref{Fig:SpectrumSmallestTenEigenvaluesVsTruncationSize},~\ref{Fig:SpectrumSmallestEigenvalueVsTruncationSize} and~\ref{Fig:SpectrumLargeEigenvaluesVsTruncationSize}, we can raise some general statements about the spectra of $\hat{H}^{(N)}_n$ in the limit $N\rightarrow\infty$.

The case $n=1$ is quite straightforward: the eigenvalues have a constant spacing that is proportional to $N^{-1/2}$. The spectrum therefore approaches a continuous spectrum as we increase $N$; in the infinite-$N$ limit, an infinite number of eigenvalues converge to zero. The case $n=2$ is the most difficult to analyse: the small eigenvalues keep decreasing throughout the range of $N$ values plotted in Fig.~\ref{Fig:SpectrumSmallestTenEigenvaluesVsTruncationSize}, without a clear indication that the level spacing approaches a finite value. This result suggests that the spectrum will be continuous in the limit of infinite $N$. On the other hand, the asymptotic value obtained from fitting the data in Fig.~\ref{Fig:SpectrumSmallestEigenvalueVsTruncationSize} is nonzero. However, the asymptotic value obtained from the fitting is far outside the range of the plotted data points, so that we cannot consider this asymptotic value reliable. As already mentioned above, $\hat{H}_2$ is self-adjoint and has a continuous spectrum, and numerical simulations with truncated Hamiltonians are known to accurately predict squeezing dynamics, provided the number of photons remains sufficiently small compared to the dimension of the state space in the simulations. 

The cases $n=3$ and $n=4$ seem to behave better. In both cases, the small eigenvalues converge quickly to finite values, indicating that the spectrum for even $N$ converges and is discrete in the infinite-$N$ limit. Likewise, the spectrum for odd $N$ converges and is discrete for $N\rightarrow\infty$ but differs from the even case. A similar behaviour is known from other operators with a continuous set of self-adjoint extensions \cite{reed1}. Similarly, the cases $n\geq 5$, which are not shown in the plots, give discrete spectra which differ for even and odd $N$.

\section{Mathematical mechanisms behind the numerical effects}
\label{Sec:Math}

The numerical results obtained throughout this paper originate in the properties of unbounded operators in an infinite-dimensional Hilbert space. To understand them, we will first briefly revisit some fundamental notions concerning these operators, and in particular, how they behave under finite-dimensional truncations. We will then proceed to show how these considerations will directly allow us to interpret the numerical results shown in the previous section.

\subsection{Self-adjointness and convergence of numerical simulations}

Let \( \hat U(t) \) be the unitary evolution operator of a quantum system on a Hilbert space \(\mathcal{H}\). Formally, the Hamiltonian \( \hat H \) is defined as the generator of this evolution,
\begin{equation}
\label{eq:generator}
\hat H = i \frac{d}{dt} \hat U(t)\Big|_{t=0}. 
\end{equation}
In finite dimensions, \( \hat H \) is simply a Hermitian matrix, and the derivative exists on the entire Hilbert space. However, in infinite-dimensional Hilbert spaces, such as the space of square-integrable wavefunctions on the real line \( L^2(\mathbb{R}) \), the situation is more delicate. Many physically relevant operators are unbounded, meaning they are not defined on all vectors in the Hilbert space. For example, the position operator \( \hat x \) acts as multiplication by \( x \), but \( x \psi \notin L^2(\mathbb{R}) \) for all \( \psi \in L^2(\mathbb{R}) \). A concrete example is the Cauchy distribution \( \psi(x) = \frac{1}{\pi(1 + x^2)} \), which lies in \( L^2(\mathbb{R}) \), while \( x \psi \notin L^2(\mathbb{R}) \). 
As a result, the Hamiltonian cannot be defined as an operator on the full Hilbert space, and one must instead restrict $\hat{H}$ to a \emph{domain} \( \mathcal{D}(\hat H) \subset \mathcal{H} \) consisting of all vectors for which the derivative in Eq.~\eqref{eq:generator} exists in \( \mathcal{H} \). An operator constructed in this way is said to be \emph{self-adjoint}.

In practice, explicitly determining the full domain of self-adjointness is often unfeasible. Therefore, one typically begins with a smaller, more manageable domain \( \mathcal{D}_0 \), and studies the restriction of \( \hat H \) to it. In favourable cases, this is enough to determine the full dynamics: there exists a unique evolution \( \hat U(t) \) compatible with the action of \( \hat H \) on \( \mathcal{D}_0 \). When this occurs, we say that \( \hat H \) is \emph{essentially self-adjoint} on \( \mathcal{D}_0 \), and that \( \mathcal{D}_0 \) is a \emph{core} for \( \hat H \). However, this is not always the case: the chosen domain \( \mathcal{D}_0 \) may be ``too small'', in the sense that it allows for multiple distinct evolution groups (see Fig. \ref{fig:self-adjoint}). These correspond to different self-adjoint extensions of \( \hat H \), which may share the same formal expression but differ in their domains, and hence encode distinct physical behaviour, such as different boundary conditions on the actually realized states.

\begin{figure}[t]
\centering
\begin{tikzpicture}[scale=1, every node/.style={scale=1}]

\definecolor{myblue}{RGB}{0,70,180}
\definecolor{mygreen}{RGB}{0,120,60}
\definecolor{myred}{RGB}{200,0,0}

\draw[thick,myred] (0,0) circle (1cm) node {$\hat H,  \mathcal{D}_0$};
\draw[thick,myblue,rotate=30] (-0.6,-0.1) ellipse (1.7cm and 1.3cm) node[xshift=-0.15cm,yshift=-0.9cm] {$\hat H_2,  \mathcal{D}_2$};
\draw[thick,mygreen,rotate=-30] (-0.6,0.1) ellipse (1.7cm and 1.3cm) node[xshift=-0.15cm,yshift=0.9cm] {$\hat H_1,  \mathcal{D}_1$};

\node at (-3.5cm,0.8cm) {$\hat U_1(t)$};
\draw[thick,<->] (-3cm,0.8cm) -- (-1.7cm,1cm);

\node at (-3.5cm,-0.8cm) {$\hat U_2(t)$};
\draw[thick,<->] (-3cm,-0.8cm) -- (-1.7cm,-1cm);

\end{tikzpicture}
\caption{Pictorial representation of a non–essentially self-adjoint operator $\hat H$ on $\domain_0$. Two distinct self-adjoint extensions $\hat H_1$ and $\hat H_2$ (with dense domains $\domain_1$ and $\domain_2$) both extend $\hat H$ but yield different time evolutions $\hat U_1(t)$ and $\hat U_2(t)$. The diagram is purely schematic, as dense subspaces cannot be represented faithfully in this way.}
\label{fig:self-adjoint}
\end{figure}

Let us now explain how these considerations are relevant to the present work. From a numerical perspective, it is important to recognize that simulations are performed in a proper subspace \( \mathcal{D}_0 \subset \mathcal{H} \) spanned by \textit{finite} linear combinations of elements of a given basis of \(\mathcal{H}\) (typically, the eigenfunctions of a reference Hamiltonian), while the full Hilbert space is spanned by possibly infinite linear combinations. Regardless of how large the number $N$ of eigenfunctions considered for a given state in $\domain_0$ is, we are always operating within a bounded, finite-dimensional approximation. This has direct implications for interpreting numerical simulations involving unbounded operators: the validity of finite-dimensional truncations in the limit $N\to\infty$ is not ensured, and must be assessed \textit{a priori} \cite{fischer-wrong-box,arzani2025effective} or \textit{a posteriori} \cite{etienney2025aposteriori}. The scenario is highly dependent on whether $\hat{H}$ is essentially self-adjoint on $\domain_0$ or not:
\begin{itemize}
    \item[(i)] If \( \hat H \) is essentially self-adjoint on \(\mathcal{D}_0\), numerical simulations performed via finite-dimensional truncations generically converge to the exact dynamics generated by the unique self-adjoint extension of \( \hat H \);
    \item[(ii)] If \( \hat H \) is not essentially self-adjoint on \(\mathcal{D}_0\), there are multiple dynamics that are compatible, in principle, with numerical simulations. In particular:
    \begin{itemize}
        \item If $\hat{H}$ is bounded from below, i.e. it has a finite ground state energy, then numerical simulations will always converge to the dynamics generated by a specific self-adjoint extension of $\hat{H}$, corresponding to the so-called Friedrichs extension~\cite{fischer-wrong-box}; 
        \item If $\hat{H}$ is not bounded from below, it is not clear \textit{a priori} whether numerical simulations converge at all, and one needs to examine the situation on a case-by-case basis.
    \end{itemize}
\end{itemize}
In the present work, for all operators involved in the simulations, we made the implicit choice
\begin{equation}
\mathcal{D}_0 = \mathrm{Span}\{|l\rangle:l\in\mathbb{N}\},
\end{equation}
that is, $\domain_0$ is the space of finite linear combinations of Fock states, which consists precisely of states $\psi \in \hilbert$ such that there is a finite $L \in \mathbb{N}$ and a finite sequence of complex numbers $(c_l)_{l \leq L}$ with
\begin{equation}
    \psi = \sum_{l = 1}^L c_l |l\rangle \, .
\end{equation}
It is important to stress that $\mathcal{D}_0$ is not the whole Hilbert space $\hilbert$ --- states which require $L = \infty$ are missing --- but some smaller, dense, subspace. Common choices of boson Hamiltonians like the number operator \( \hat a^\dagger \hat a \) and the second-order squeezing operator \( \hat H_2 = i(\hat a^\dagger)^2 - i \hat a^2 \) are essentially self-adjoint on $\mathcal{D}_0$, and therefore their finite-dimensional approximations reproduce the correct dynamics in the large-\( N \) limit.

We will show in the remainder of this section that the numerical results observed in this paper can indeed be explained in terms of the essential self-adjointness---or lack thereof---of the operators on the space $\domain_0$. Namely:
\begin{itemize}
    \item The higher-order squeezing Hamiltonian $\hat{H}_n$ with $n\geq3$ is not essentially self-adjoint on $\domain_0$ nor is bounded from below, thus explaining the irregular behaviour observed in the simulations.     
    Still, the even--odd effect observed in the simulations can be precisely explained: finite-dimensional approximations oscillate between the dynamics of \textit{two} distinct self-adjoint extensions. We show this in Section \ref{Sec:Math_noKerr}.
    \item However, when adding a properly chosen regularizing term such as a quadratic Kerr term, one can restore essential self-adjointness and thus regularity of the numerical simulations, as discussed in Section \ref{Sec:Math_Kerr}.
\end{itemize}
The results about self-adjoint extensions building the backbone of these facts are based on~\cite{FischerInPreparation}, where the essential self-adjointness---or lack thereof---of a class of Hamiltonians, including the ones analysed in this paper, is characterized. Here, we showcase how these abstract mathematical results have strict consequences on the present numerical analysis, and prove the convergence results stated above.

\subsection{Non-regularized higher-order squeezing}\label{Sec:Math_noKerr}

We begin by considering the higher-order squeezing operator $\hat{H}_n$ with $n\geq3$, cf. Eq.~\eqref{Eq:Hamiltonian_n}, without any regularizing term. 
This case falls in the class of operators studied in~\cite[Section 4]{FischerInPreparation}, which are not essentially self-adjoint. 
Indeed, $\hat{H}_n$ admits infinitely many self-adjoint extensions parametrized by $n\times n$ unitary matrices. 
Here we will focus on two particular extensions that are sufficient to explain the numerical patterns we observe:
\begin{proposition}[\cite{FischerInPreparation}]
    \label{prop:extensions}
    Let $n \geq 3$ and $\hat{H}_n$ be defined as in~\eqref{Eq:Hamiltonian_n}, with domain $\domain_0$,
    Then $\hat{H}_n$ is not essentially self-adjoint. In particular, there exist two distinct, essentially self-adjoint extensions $\Hodd,\Heven$ of $\hat{H}_n$, having domains
    \begin{align}
        \domain(\Hodd) & = \left\{ \psi_0 + \sum_{i=0}^{n-1}c_i \sum_{j=0}^{\infty} d^{(i)}_{2j} \ket{i+2j n} \, :\, \psi_0 \in \domain_0 , \,(c_i)_{i=0}^{n-1} \in \cnum^n \right\}\,   , \\
        \domain(\Heven) & = \left\{ \psi_0 + \sum_{i=0}^{n-1}c_i \sum_{j=0}^{\infty} d^{(i)}_{2j+1} \ket{i+(2j+1) n} \, :\, \psi_0 \in \domain_0 ,\, (c_i)_{i=0}^{n-1} \in \cnum^n \right\}\,  , 
    \end{align}
    where the $d_j^{(i)}$ are suitably chosen real, positive coefficients \cite{Footnote}. Besides, for every $\psi_0\in\domain_0$,
    \begin{equation}
        \Heven\psi_0=\Hodd\psi_0=\hat{H}_n\psi_0=i \left[ \left(\hat{a}^\dagger\right)^n - \hat{a}^n \right]\psi_0.
    \end{equation}
\end{proposition}
In plain words: one can render $\hat{H}_n$ essentially self-adjoint---and thus, uniquely generating a proper unitary evolution $\hat{U}(t)$---by enlarging its domain instead of changing its expression (e.g., adding a regularizing term). Precisely, we can do so by adding a family of vectors to the domain in which either the odd or the even (times $n$) entries are 0, and the remaining entries are predetermined. Those two choices will correspond to two distinct operators, generating \textit{distinct} dynamics $\e^{-i \hat{H}_{n,\rm even}t}$, $\e^{-i \hat{H}_{n,\rm odd}t}$.

What we did in the present work was to consider finite-dimensional truncations of squeezing operators---and thus, of their dynamics---in the Fock basis $(\ket{l})_{l \in \nnum}$. Mathematically, let $\hat{P}_M$ be the projection onto the first $M$ Fock states, i.e. 
\begin{equation}\label{eq:proj}
    \hat{P}_M = \sum_{l=0}^{M-1}\ket{l}\!\bra{l} \, .
\end{equation}
Then, in our simulations, we considered the dynamics generated by the bounded operators $\hat{H}_{n,M} = \hat{P}_M \hat{H}_n \hat{P}_M$ for some large---but finite---$M$. 
In the following we prove our main result, namely that in the limit of large $M$, the dynamics converge either to $\e^{-i \hat{H}_{n,\rm even}t}$ or $\e^{-i \hat{H}_{n,\rm odd}t}$, depending on whether the truncation $M$ is taken to be even or odd:
\begin{theorem}
\label{thm:convergence}
    Let $\hat{H}_{n,M} = \hat{P}_M \hat{H}_n \hat{P}_M$ be the $M$-dimensional truncation of $\hat{H}_n$, and let $\hat{U}_M(t) = \e^{-i \hat{H}_{n,M} t}$ be the associated time evolution.
    Then, for every $t\in\mathbb{R}$,
    \begin{align}
        \lim_{j \to \infty} \hat{U}_{2j n}(t) \psi & = \e^{-i \Heven t} \psi \quad \forall \psi \in \hilbert \, , \\
        \lim_{j \to \infty} \hat{U}_{(2j+1) n}(t) \psi & = \e^{-i \Hodd t} \psi \quad \forall \psi \in \hilbert \, .
    \end{align}
\end{theorem}
\begin{proof}
    We will prove the statement concerning $\Heven$; the proof for $\Hodd$ is analogous. To this end, we  will show that, for all $\psi \in \domain(\Heven)$, 
    \begin{equation}\label{eq:claim1}
        \lim_{j \to \infty} \hat{P}_{2jn }\hat{H}_n \hat{P}_{2jn} \psi = \Heven \psi \, ,
    \end{equation}
    with $\hat{P}_{2jn}$ being the projector over the first $2jn$ Fock states, cf. Eq.~\eqref{eq:proj}. This will prove convergence of the associated unitary groups by virtue of the Kato approximation theorem~\cite[Theorem 4.8]{nagel-oneparametersemigroups-1999}.
    
    Let $\psi \in \domain(\Heven)$. By Prop.~\ref{prop:extensions}, there exists $\psi_0 \in \domain_0$ and $\psieven$ such that $\psi = \psi_0 + \psieven$, where
    \begin{itemize}
        \item $\psi_0\in\domain_0$, i.e.~it is a \textit{finite} linear combination of Fock states: there exists $L \in \nnum$ and coefficients $a_1,\dots,a_L$ such that
        \begin{equation}\label{eq:psi0}
            \psi_0 = \sum_{l = 0}^{L-1}a_l \ket{l}.
        \end{equation}
        Furthermore, $\hat{H}_n\psi_0=\Heven\psi_0$.
        \item $\psieven$ is given by a generally \textit{infinite} linear combination of Fock states,
    \begin{equation}\label{eq:psieven}
        \psieven = \sum_{l = 0}^{\infty} c_l \ket{l}\ ,
    \end{equation}
    for some suitable coefficients $c_l$ satisfying the following property: the coefficients are $0$ for $l = i+2jn$ with $0\leq i < n$ and $j \in \nnum$, while they are generally nonzero for $l = i +(2j+1)n$. 
    \end{itemize}
    We will prove that the actions of the operators $\hat{P}_{2jn}\hat{H}_n\hat{P}_{2jn}$ and $\Heven$ coincide, in the limit $j\to\infty$, separately on $\psi_0$ and $\psieven$; that is, we claim
    \begin{align}\label{eq:claim2}
         \lim_{j \to \infty} \hat{P}_{2jn }\hat{H}_n \hat{P}_{2jn} \psi_0 &= \Heven \psi_0 \, ,\\ \label{eq:claim3}
         \lim_{j \to \infty} \hat{P}_{2jn }\hat{H}_n \hat{P}_{2jn} \psieven &= \Heven \psieven \, ,
    \end{align}
    which, by linearity, imply Eq.~\eqref{eq:claim1}.
    
    Let us begin by proving Eq.~\eqref{eq:claim2}. Since $\psi_0$ is a combination of finitely many Fock states $\ket{0},\ket{1},\ldots,\ket{L-1}$ (cf. Eq.~\eqref{eq:psi0}), we have $\hat{P}_{L}\psi_0=\psi_0$. Besides, as $\hat{H}_n$ can create at most $n$ photon excitations, we will also have $\hat{P}_{L+n}\hat{H}_n\psi_0=\hat{H}_n\psi_0$. Therefore,
    \begin{equation}
        \hat{H}_{n,L+n} \psi_0 = \hat{P}_{L+n} \hat{H}_n \hat{P}_{L + n} \psi_0 = \hat{P}_{L+n} \hat{H}_n \psi_0 = \hat{H}_n \psi_0 \, =\Heven\psi_0
    \end{equation}
    as $\Heven\psi_0=\hat{H}_n\psi_0$;  whence, \textit{a fortiori}, 
\begin{equation}
    \lim_{j \to \infty} \hat{P}_{2jn }\hat{H}_n \hat{P}_{2jn} \psi_0 = \hat{P}_{L+n} \hat{H}_n \hat{P}_{L + n} \psi_0= \Heven \psi_0,
\end{equation}
since $2jn\geq L+n$ for $j$ large enough. We thus proved Eq.~\eqref{eq:claim2}.

We now proceed to the proof of Eq.~\eqref{eq:claim3}. To this end, we note that $\braket{m}{ (\hat{a}^\dagger)^n| l} = 0$ for $l > m$ and $\braket{m}{ \hat{a}^n |l} = 0$ for $l > m+n$.
    Thus, expanding $\psieven$ as in Eq.~\eqref{eq:psieven},
    and noting that $1-\hat{P}_{2jn}=\sum_{l=2jn}^\infty\ket{l}\!\bra{l}$, we have

    \begin{align}
        \label{proofeq:convergence}
        \hat{P}_{2jn} \Heven (1-\hat{P}_{2jn}) \psieven & = \sum_{m = 0}^{2jn-1}\sum_{l=2jn}^\infty \ket{m}\!\bra{m}\Heven\ket{l}\!\braket{l}{\psieven}\\ 
                \label{proofeq:convergence-second}
        & = \sum_{m = 0}^{2jn-1} \sum_{l = 2jn}^{\infty} c_l \ket{m}\!\braket{m}{ \hat{H}_n|l }\\
                \label{proofeq:convergence-third}
        & = \sum_{m = 0}^{2jn-1} \sum_{l = 2jn}^{\infty} c_l \ket{m}\!\braket{m}{ (i (\hat{a}^\dagger)^n - i \hat{a}^n)|l }\\
        \label{proofeq:convergence-adagger}& =  -i \sum_{m = 0}^{2jn-1} \sum_{l = 2jn}^{\infty} c_l \ket{m}\!\braket{m}{  \hat{a}^n |l } \\
        \label{proofeq:convergence-a}& = -i \sum_{m = 0}^{2jn-1} \sum_{l = 2jn}^{2jn+n-1} c_l \ket{m}\!\braket{m}{  \hat{a}^n |l}\, , 
    \end{align}
where in Eq.~\eqref{proofeq:convergence} we used the explicit expression of the projectors $\hat{P}_{2jn}$ and $1-\hat{P}_{2jn}$; in Eq.~\eqref{proofeq:convergence-second}, we used the fact that $\hat{H}_{n,\rm even}$ acts as $\hat{H}_n$ on all vectors in $\domain_0$ (thus, \textit{a fortiori}, on all Fock states) and the equality $\braket{l}{\psieven}=c_l$ which comes from Eq.~\eqref{eq:psieven}; in Eq~\eqref{proofeq:convergence-adagger}, the fact that, since $l > m$, all terms $\braket{m}{(\hat{a}^\dag)^n|l}$ vanish; and finally, in Eq.~\eqref{proofeq:convergence-a}, the equality $\braket{m}{(a)^n|l}=0$ whenever $l>m+n$.

We thus see that the only coefficients $c_l$ appearing in Eq.~\eqref{proofeq:convergence-a} are those with $2jn\leq l\leq 2jn+n-1$. But these coefficients are precisely those that can be written as $c_{i+2jn}$ with $0 \leq i < n$, which are zero. Therefore, we proved
\begin{equation}
     \hat{P}_{2jn} \Heven (1-\hat{P}_{2jn}) \psieven=0,
\end{equation}
whence
\begin{equation}\label{eq:limit1}
     \hat{P}_{2jn} \Heven \hat{P}_{2jn}\psieven=\hat{P}_{2jn} \Heven (\hat{P}_{2jn}-1+1)\psieven=\hat{P}_{2jn}\Heven\psieven.
\end{equation}
On the other hand,
\begin{equation}\label{eq:limit2}
     \hat{P}_{2jn} \Heven \hat{P}_{2jn}\psieven=\hat{P}_{2jn}\hat{H}_n \hat{P}_{2jn}\psieven,
\end{equation}
as $\hat{P}_{2jn}\psieven\in\domain_0$ (it is a finite linear combination of Fock states), and again $\hat{H}_n$ and $\Heven$ coincide on such vectors. Thus, combining Eqs.~\eqref{eq:limit1}--\eqref{eq:limit2}, we finally have
\begin{align}
    \lim_{j\to\infty}\hat{P}_{2jn} \hat{H}_n \hat{P}_{2jn}\psieven&=\lim_{j\to\infty}\hat{P}_{2jn} \Heven \hat{P}_{2jn}\psieven\\&=\lim_{j\to\infty}\hat{P}_{2jn}\Heven\psieven=\Heven\psieven,
\end{align}
since $\hat{P}_{2jn}$ converges strongly to the identity. We thus proved Eq.~\eqref{eq:claim3} as well.

Having proven both Eqs.~\eqref{eq:claim2}--\eqref{eq:claim3}, we proved Eq.~\eqref{eq:claim1}: for every $\psi\in\domain(\Heven)$, $\hat{P}_{2jn }\hat{H}_n \hat{P}_{2jn}\psi$ converges to $\Heven\psi$. As anticipated, we can now invoke the Kato approximation theorem~\cite[Theorem 4.8]{nagel-oneparametersemigroups-1999}: as all operators $\hat{P}_{2jn }\hat{H}_n \hat{P}_{2jn}$, being bounded and self-adjoint, admit the dense domain $\domain(\Heven)$ as a core, and they converge strongly to $\Heven$ on this core, we have
    \begin{equation}
        \lim_{j \to \infty} \e^{-i \hat{P}_{2jn}\hat{H}_n \hat{P}_{2jn} t} \psi = \e^{-i \Heven t} \psi \quad \forall \psi \in \hilbert \, ,
    \end{equation}
    thus concluding the proof.
 \end{proof}
 
 We thus showed that the dynamics of the truncated Hamiltonian $\hat{H}_{n,M}$ converges to different self-adjoint extensions of $\hat{H}_n$ depending on whether we increase $M$ in even multiples of $n$ ($M = 2jn$) or odd multiples of $n$ ($M=2jn+n$).
 In Section~\ref{Sec:Theory}, we noted that the dynamics of the vacuum state $\ket{0}$ is confined to the subspace $\{\ket{0},\ket{n},\ket{2n},\dots\}$, hence we chose as an approximation space the $N$ dimensional subspace $\{\ket{0},\ket{n},\ket{2n},\dots,\ket{n \times(N-1)}$.
 We then analysed the time evolution $\hat{U}_{Nn}(t)\ket{0} = \e^{-i \hat{H}_{n,nN} t}\ket{0} $ of the vacuum state.
By Theorem~\ref{thm:convergence}, we can now analytically conclude that these approximations converge to different time evolutions for even and odd $N$:
\begin{align}
   \lim_{N \to \infty, \,N \text{ even}} \e^{-i \hat{H}_{n,N n} t}\ket{0} & =  \lim_{j \to \infty} \e^{-i \hat{H}_{n,2j n} t}\ket{0}  = \e^{-i \Heven t}\ket{0} \, , \\
    \lim_{N \to \infty, \,N \text{ odd}} \e^{-i \hat{H}_{n,N n} t}\ket{0} & =  \lim_{j \to \infty} \e^{-i \hat{H}_{n,(2j+1) n} t}\ket{0}   = \e^{-i \Hodd t}\ket{0} \, . 
\end{align}
In particular, both limiting time evolutions are physical, in the sense that they correspond to self-adjoint operators---but distinct ones.
In an actual experiment, the physics at hand would encode the domain, and thus the ``real'' self-adjoint extension, similar to how boundary conditions determine domains in different systems.  
As the observations in Section~\ref{Sec:Simulations} appear for both the even and the odd limit, we can conclude that they are not just numerical artifacts, but profound features of the two self-adjoint extensions of $\hat{H}_n$ corresponding to the two truncation schemes. However, it is unlikely that an actual experimental realization corresponds to one or the other of these extensions because the respective restrictions on the states are too artificial. There are many more extensions which have not been analyzed here and determine unitary time evolutions (respectively squeezing operations) as well. From a physical viewpoint it is more feasible to identify the full operator  corresponding to the actual physics and check whether it is essentially self-adjoint as we do in the following section.

Finally, let us remark that the peculiar effect observed here---convergence to different self-adjoint extensions for even and odd $N$---can only happen because the operator $\hat{H}_n$ is neither bounded from below nor from above.
If it were bounded from below, the time evolution would converge to a unique self-adjoint extension, the so--called Friedrichs extension for all $N$~\cite{fischer-wrong-box}, so that no irregular behaviour would be directly detected.

\subsection{Regularized higher-order squeezing operators}\label{Sec:Math_Kerr}

Since the irregular behavior of simulations for $\hat{H}_n$ observed in this paper can be explained in terms of the lack of essential self-adjointness of the Hamiltonian, it is natural to guess that, on the other hand, the absence of such irregularities in the presence of additional Kerr terms arises from the fact that such terms make the Hamiltonian essentially self-adjoint.

Indeed, the following result was stated in \cite[Section 5]{FischerInPreparation}, see Example 5.3:

\begin{proposition}[\cite{FischerInPreparation}]\label{prop:esssa}
    Let $n\geq3$, $h\in\mathbb{N}$, $K>0$, and $\hat{H}_{n,h\,\rm Kerr}$ the operator with domain $\domain_0$ defined by
    \begin{equation}
        \hat{H}_{n,h\,\rm Kerr}=i\left[(\hat{a}^\dag)^n-\hat{a}^n\right]+K(\hat{a}^\dag)^h\hat{a}^h.
    \end{equation}
    Then:
    \begin{itemize}
        \item If $n>2h$, $\hat{H}_{n,h,\rm Kerr}$ is not essentially self-adjoint;
        \item If $n<2h$, $\hat{H}_{n,h,\rm Kerr}$ is essentially self-adjoint;
        \item If $n=2h$, $\hat{H}_{n,h,\rm Kerr}$ is essentially self-adjoint when $K>2$, and is not essentially self-adjoint when $K<2$.
    \end{itemize}
\end{proposition}
In simple terms, $\hat{H}_{n,h,\rm Kerr}$ is or is not essentially self-adjoint depending on which of the two components is of higher order in the photon operators (which correspond precisely to $n$ and $2h$), and thus dominates in the large-photon limit. If $n=2h$, then the relative strengths between the two terms---with the notation of this paper, the Kerr coefficient $K$---becomes the relevant quantity.

In particular, we can use these results on essential self-adjointness to prove that numerical simulations of higher-order squeezing operators converge, provided that a dominant Kerr term is added for regularization.
The simulation results in this paper are in alignment with this picture.
Specifically, our simulations for the Hamiltonians \eqref{Eq:Hamiltonian_3_QuadraticKerr} ($n=3$, $h=2$), \eqref{Eq:Hamiltonian_3_QuarticKerr} ($n=3$, $h=4$), and \eqref{Eq:Hamiltonian_4_QuarticKerr} ($n=4,h=4$) showed that the Kerr term becomes eventually effective in regulating the dynamics with respect to the truncation dimension $N$. As in all three cases $2h>n$, this can be now explained by virtue of Proposition \ref{prop:esssa}: these Hamiltonians are essentially self-adjoint, so that finite-dimensional truncations will eventually converge to the actual evolution. In contrast, the simulations for the Hamiltonian \eqref{Eq:Hamiltonian_4_QuadraticKerr} ($n=4,h=2$) showed convergence in the truncation dimension only for $K>2$, without any even--odd effect, while for $K<2$ dimension-dependent effects appeared. Again, this can now be explained in mathematical terms by Proposition \ref{prop:esssa}: as $2h=n$, the Hamiltonian is only essentially self-adjoint for $K>2$, with $K=2$ precisely being the critical point where essential self-adjointness breaks down.

\section{Conclusion}
\label{Sec:Conclusion}

The mathematical modelling and numerical simulation of generalized squeezing have proven to be challenging tasks with many potential pitfalls for higher-order squeezing with $n\geq 3$. The latest of these pitfalls is the dependence of the dynamics and the spectrum on the parity of the finite-dimensional truncations. After demonstrating the appearance of this parity dependence with a few examples, we analysed the photon number dynamics and the spectrum of the truncated squeezing Hamiltonian to gain some insight into the cause of this dependence. Making use of the results developed by some of us in \cite{FischerInPreparation}, we traced back the mathematical origin of this phenomena to the fact that the squeezing operator $\hat{H}_n$ is essentially self-adjoint on Fock states for $n=1,2$, a prerequisite for unique dynamics, but \textit{not} for $n\ge 3$. Nevertheless, as we also demonstrated, the dynamics of the even and odd truncation schemes correspond separately to well-defined unitary evolutions associated with two different self-adjoint \emph{extensions} of $\hat{H}_n$. 

As both schemes entail boundary conditions on the wave functions that might be regarded as unphysical, we have also analysed how the situation changes by introducing additional terms in the Hamiltonian. These or similar ones will undoubtedly be present in a realistic physical system. We show that, indeed, such terms can regulate the dynamics by creating natural cutoffs of the photon content and lead to effective Hamiltonians that are self-adjoint and allow thus valid extrapolations of the numerical simulation from finite-dimensional spaces to infinite dimension. Again, we provided a mathematical explanation of this phenomenon: additional field terms can restore the essential self-adjointness of the Hamiltonian, so that their finite-dimensional truncations produce a unique, truncation-dependent evolution in the infinite-size limit.

Our results shed light on the mathematical modelling and numerical simulation of nonlinear quantum optics phenomena that we expect to be relevant more broadly than the specific problem of generalized squeezing treated in this work.

\section*{Acknowledgments}

We would like to thank Mohammad Ayyash for useful discussions. We would like to thank Rub\'en Gordillo and Ricardo Puebla for sharing their manuscript \cite{Gordillo} with us prior to its publication. SA was supported by Japan's Ministry of Education, Culture, Sports, Science and Technology (MEXT) Quantum Leap Flagship Program Grant Number JPMXS0120319794. DL acknowledges financial support by Friedrich-Alexander-Universität Erlangen-Nürnberg through the funding program “Emerging Talent Initiative” (ETI), and was partially supported by the project TEC-2024/COM-84 QUITEMAD-CM. Daniel Braak acknowledges funding by the Deutsche Forschungsgemeinschaft (DFG) under grant no. 439943572.

\section*{Data availability}

The datasets generated and/or analysed during the current study are available from the corresponding author on reasonable request.

\section*{Appendix A: Additional examples of regulating the dynamics via Kerr interactions}

In Sec.~\ref{Sec:Simulations_Kerr}, we showed how a sufficiently strong Kerr interaction term in the Hamiltonian can regulate the dynamics and make the squeezing operator well behaved. In this appendix, we show a few additional examples that elucidate the role of the Kerr term in the dynamics.

\begin{figure}[h]
\includegraphics[width=8cm]{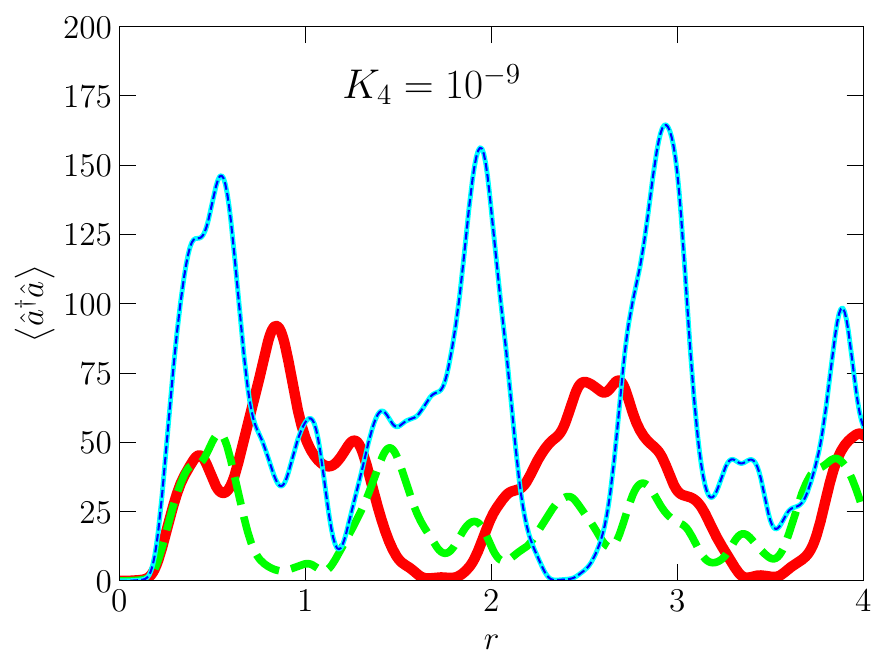}
\includegraphics[width=8cm]{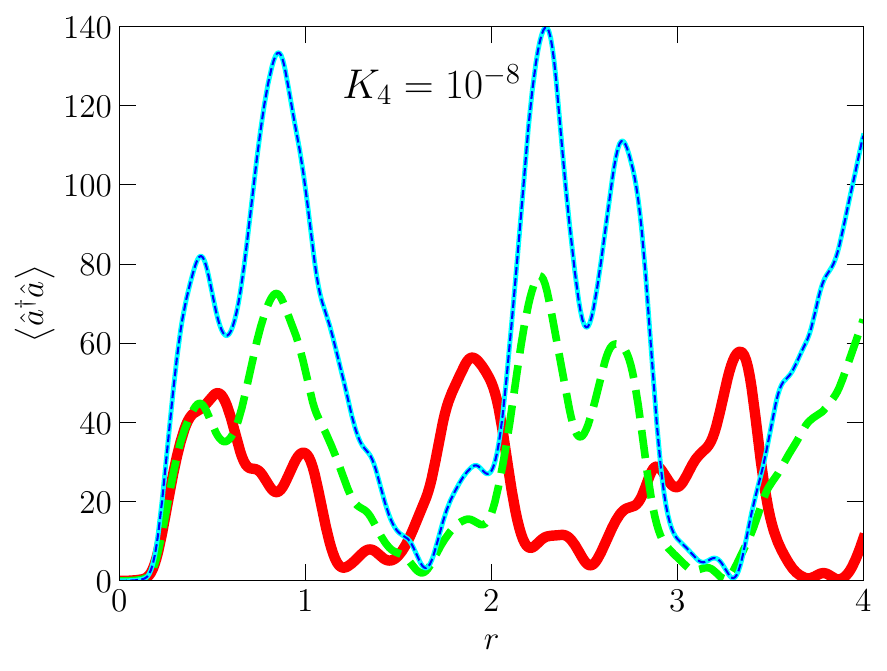}
\includegraphics[width=8cm]{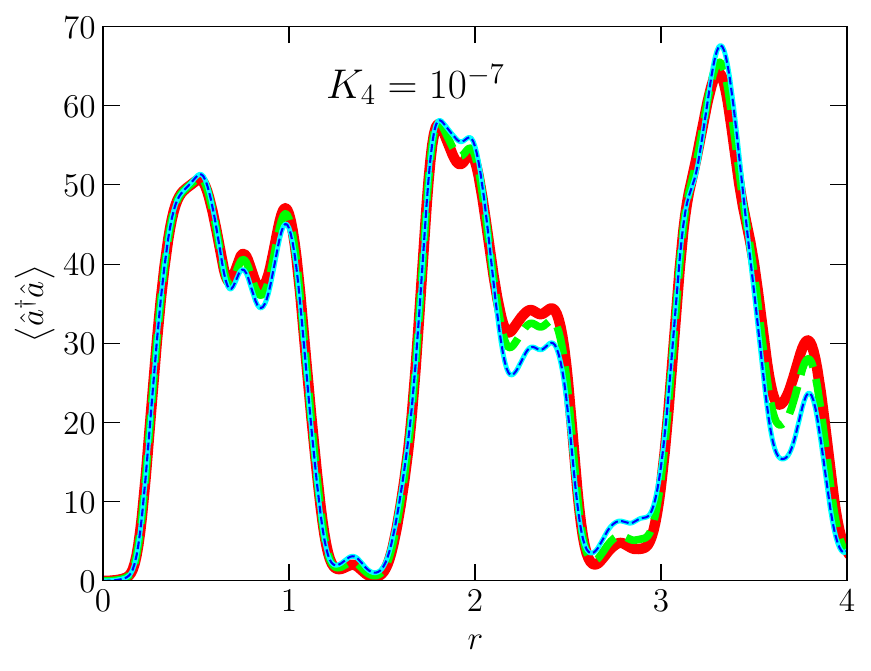}
\includegraphics[width=8cm]{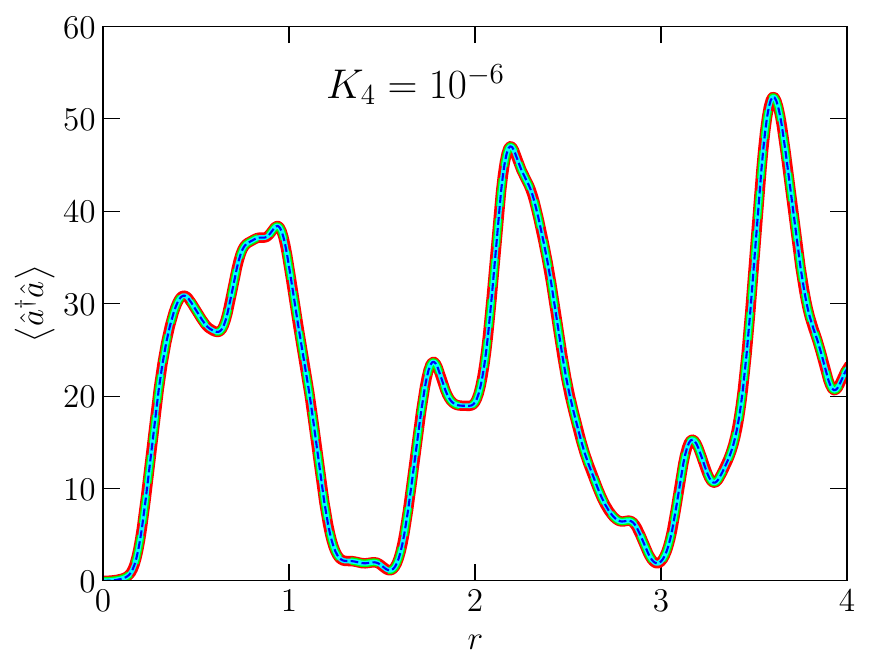}
\caption{Same as Fig.~\ref{Fig:AveragePhotonNumberVsSqueezingParameterKerr3Quadratic}, but with a quartic Kerr term with strength $K_4$.}
\label{Fig:AveragePhotonNumberVsSqueezingParameterKerr3Quartic}
\end{figure}

In addition to the example presented in Sec.~\ref{Sec:Simulations_Kerr}, the second example that we consider employs a quartic Kerr term:
\begin{equation}
\hat{H}_{3, \rm quartic \ Kerr} = i \left[ \left(\hat{a}^\dagger\right)^3 - \hat{a}^3 \right] + \frac{K_4}{4!} \left(\hat{a}^\dagger\right)^4 \hat{a}^4,
\label{Eq:Hamiltonian_3_QuarticKerr}
\end{equation}
where we have included the factor $4!$ for convenience. One motivation for considering the quartic case is the following: it allows us to make the Kerr term weaker at low photon numbers, while still ensuring that it rises rapidly at higher photon numbers to create the desired cutoff effect. Of course, the actual presence (or absence) of these terms is dictated by the effective Hamiltonian of the physical realization. The results are shown in Fig.~\ref{Fig:AveragePhotonNumberVsSqueezingParameterKerr3Quartic}. As in the quadratic case, the dynamics becomes independent of $N$ above a certain value of $K_4$. Performing a similar estimate as in the quadratic case, we find that we now must require $K_4(nN)^{5/2}/4!>1$. By taking $K_4=10^{-7}$ and $N=1000$, we have $K_4(nN)^{5/2}/4!=2$. As expected, the cutoff effect imposed by the Kerr term becomes noticeable when the Kerr term becomes comparable to the squeezing term at the upper photon number values in the truncated Hamiltonian. It is worth noting here that, at $10^{-7}$, the confining effects of the Kerr term start to take place at large photon numbers, close to $nN=3000$. As a result, the maximum height of the oscillations in Fig.~\ref{Fig:AveragePhotonNumberVsSqueezingParameterKerr3Quartic} does not directly correspond to the Kerr-term-induced energy barrier, e.g.~the point where the Kerr-term matrix elements become comparable to the squeezing matrix elements. This result is an indication that the maximum height is determined by subtler interference effects. Another way to see that is by noticing that the maximum height does not scale as $K_4^{-2/5}$ as we increase $K_4$ in Fig.~\ref{Fig:AveragePhotonNumberVsSqueezingParameterKerr3Quartic}, which is what we would expect if the oscillation amplitude were set directly by the matching of diagonal and off-diagonal matrix elements in the Hamiltonian. Furthermore, if we compare the cases $K_4=10^{-7}$ and $K_4=10^{-6}$ (and similarly, if we compare the cases $K=10^{-1}$ and $K=1$ in Fig.~\ref{Fig:AveragePhotonNumberVsSqueezingParameterKerr3Quadratic}), we see that the oscillation frequency is largely independent of $K$. Specifically, the oscillation frequency does not decrease rapidly with increasing $K$. This result emphasizes that the average photon number dynamics does not follow the simple picture where the photon number grows with some truncation-size-independent speed until it hits the cutoff point set by the Kerr coefficient, and is then reflected and decreases back to zero. Instead, the higher-order squeezing dynamics explores the entire available Hilbert space in an amount of time that is almost independent of the details of the Kerr term.

\begin{figure}[h]
\includegraphics[width=8cm]{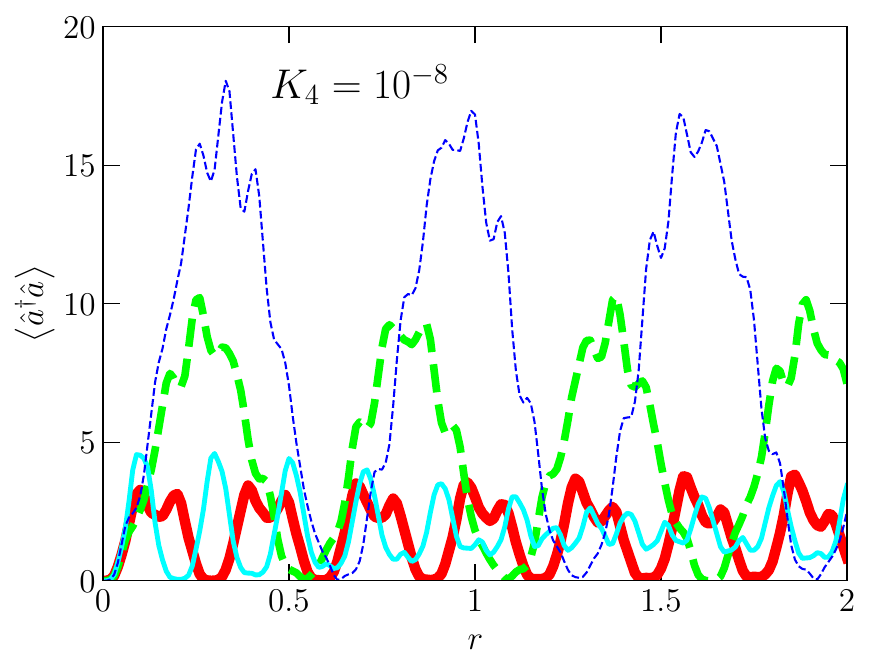}
\includegraphics[width=8cm]{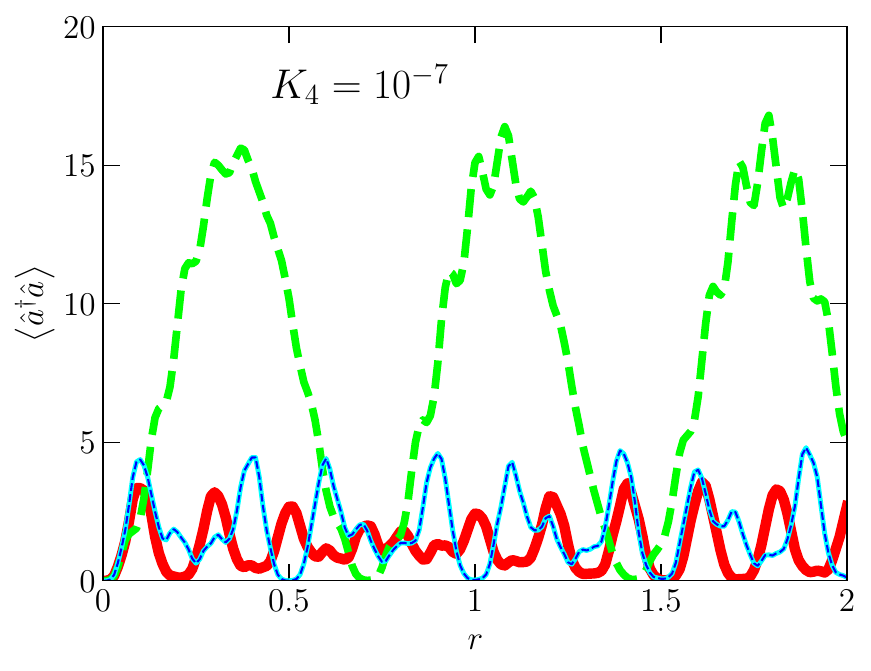}
\includegraphics[width=8cm]{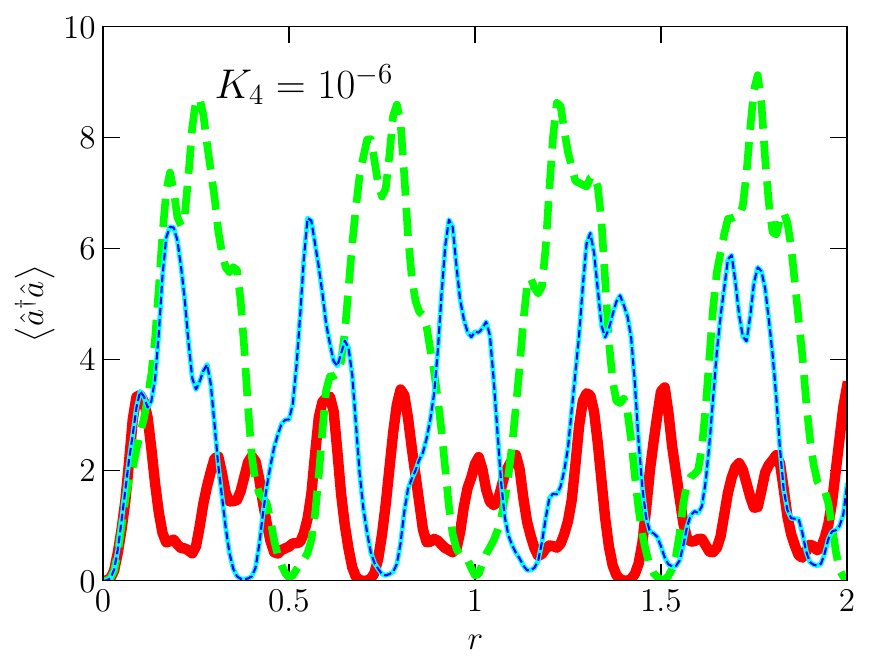}
\includegraphics[width=8cm]{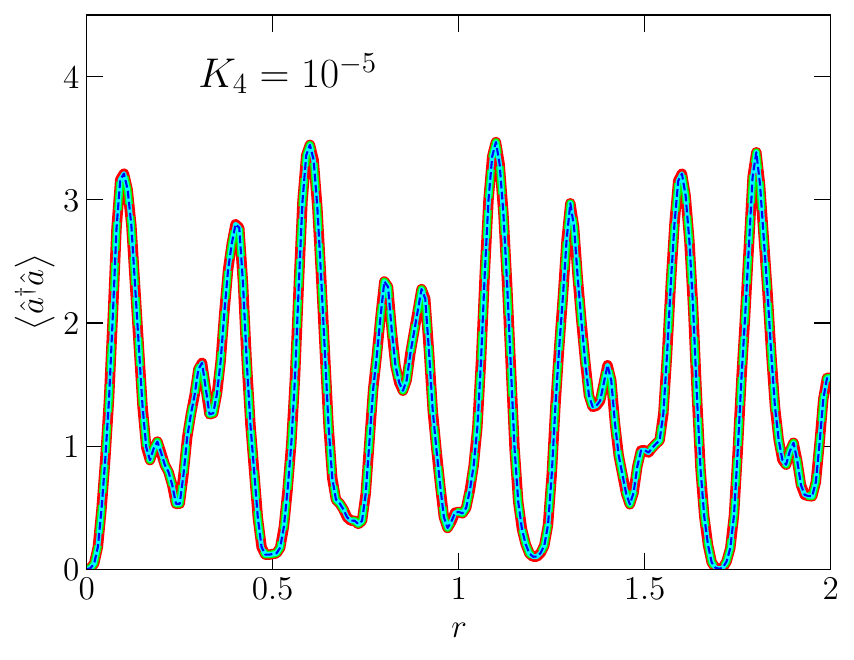}
\caption{Average photon number $\left\langle \hat{a}^{\dagger} \hat{a} \right\rangle$ for the state $\hat{U}^{(N)}_4(r)\ket{0}$ as a function of the squeezing parameter $r$ with a quartic Kerr term of varying strength ($K_4$). The red, green, cyan and blue lines correspond, respectively, to $N=1000$, 1001, $10^4$ and $10^4+1$.}
\label{Fig:AveragePhotonNumberVsSqueezingParameterKerr4Quartic}
\end{figure}

The third example that we consider is the case of quad-squeezing ($n=4$) with a quartic Kerr term:
\begin{equation}
\hat{H}_{4, \rm quartic \ Kerr} = i \left[ \left(\hat{a}^\dagger\right)^4 - \hat{a}^4 \right] + \frac{K_4}{4!} \left(\hat{a}^\dagger\right)^4 \hat{a}^4.
\label{Eq:Hamiltonian_4_QuarticKerr}
\end{equation}
Taking into account that the $n=3$ dynamics in Fig.~\ref{Fig:AveragePhotonNumberVsSqueezingParameter} are rather irregular, the case $n=4$ allows us to compare the dynamics in the presence of the Kerr term with the rather regular dynamics of the corresponding case in Fig.~\ref{Fig:AveragePhotonNumberVsSqueezingParameter}. The results with the Kerr term are plotted in Fig.~\ref{Fig:AveragePhotonNumberVsSqueezingParameterKerr4Quartic}. Following a similar argument to the one presented above for $n=3$, we find that the Kerr term becomes effective in regulating the dynamics when $K_4(nN)^2/4!>1$. If we examine the dynamics in the case $K_4=10^{-5}$, where the dynamics is independent of $N$, we find that the oscillations are quite irregular and do not resemble any of the curves in Fig.~\ref{Fig:AveragePhotonNumberVsSqueezingParameter}. 

In all three examples presented above, we found that, when the Kerr term is sufficiently strong, the results become independent of the truncation size. An alternative way to look at this result is as follows: for any finite value of the Kerr coefficient, there exists a minimum truncation size $N_{\rm min}$ such that the simulation results are independent of $N$ when $N>N_{\rm min}$. 
As we see in section \ref{Sec:Math_Kerr}, all these facts can be mathematically demonstrated. 

\begin{figure}[h]
\includegraphics[width=8cm]{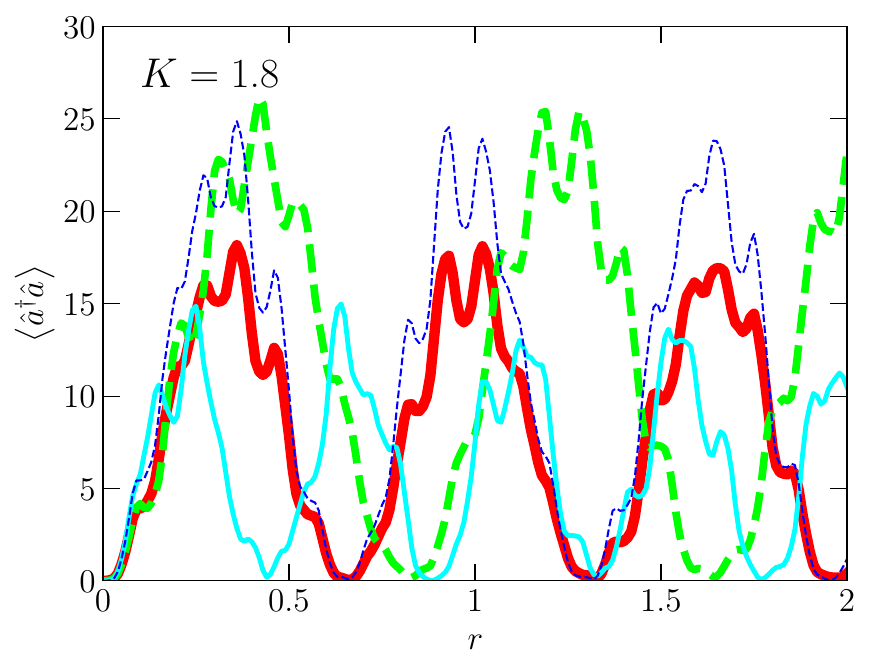}
\includegraphics[width=8cm]{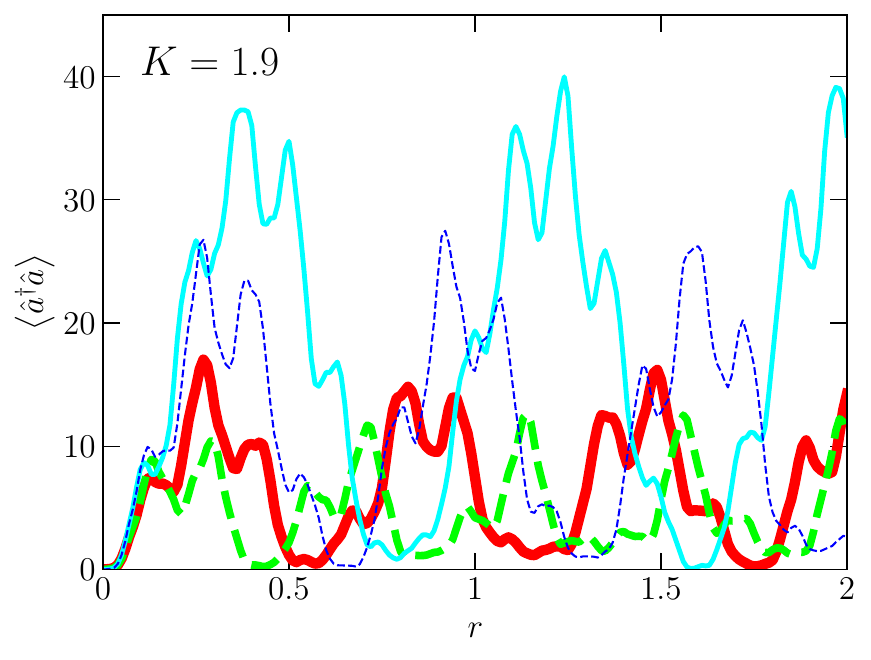}
\includegraphics[width=8cm]{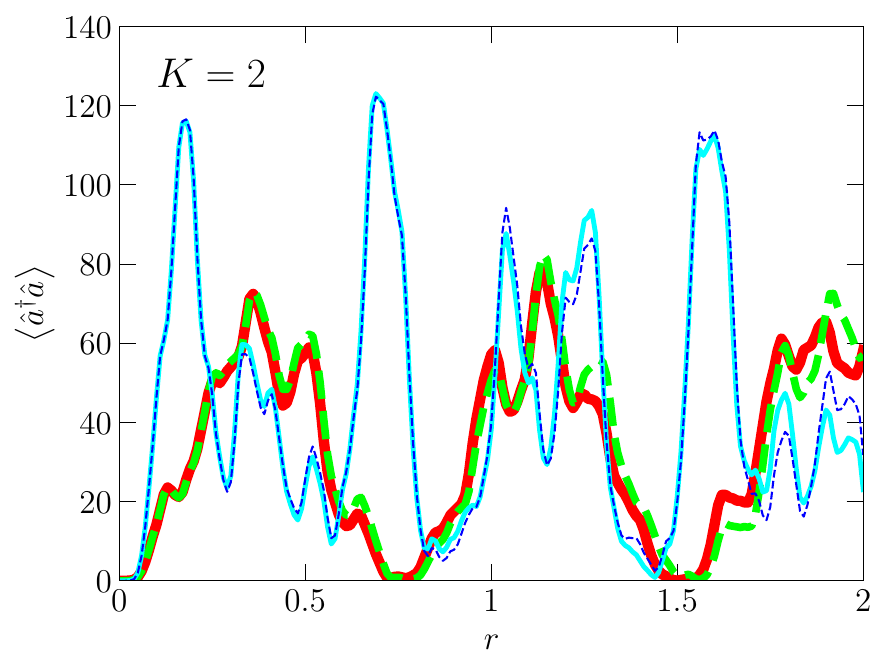}
\includegraphics[width=8cm]{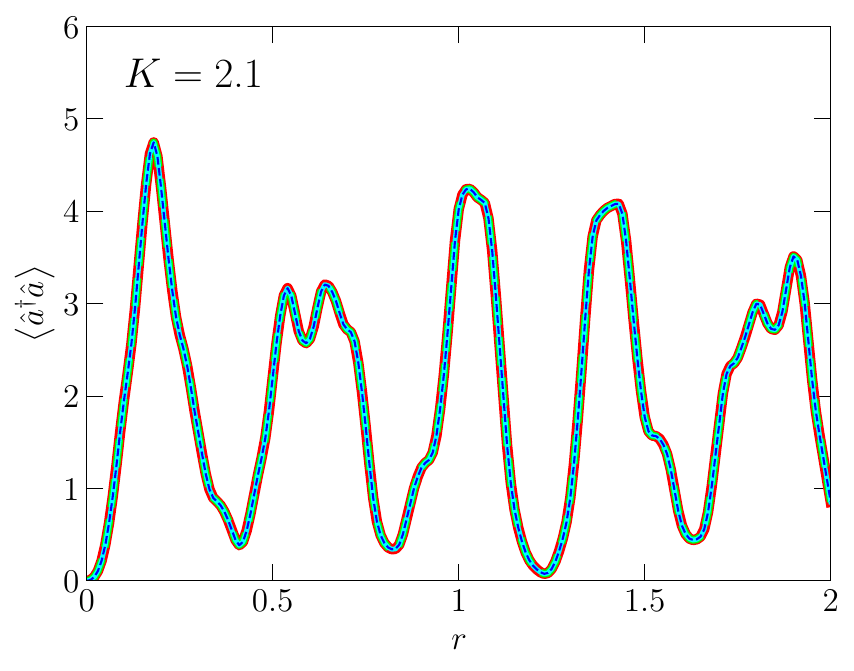}
\caption{Same as Fig.~\ref{Fig:AveragePhotonNumberVsSqueezingParameterKerr4Quartic}, but with a quadratic Kerr term. In this case, the results become independent of $N$ as soon as we cross the point $K=2$. When $K<2$, all the simulation results are generally different from each other. At $K=2$, there is a good, but not perfect, agreement between the even-$N$ simulations and the odd-$N$ simulations. Another noteworthy aspect of the case $K=2$ is that the photon number reaches values that are significantly higher than those reached for all other values of $K$.}
\label{Fig:AveragePhotonNumberVsSqueezingParameterKerr4Quadratic}
\end{figure}

The final example that we consider is the case of quad-squeezing ($n=4$) with a quadratic Kerr term:
\begin{equation}
\hat{H}_{4, \rm quadratic \ Kerr} = i \left[ \left(\hat{a}^\dagger\right)^4 - \hat{a}^4 \right] + K \left(\hat{a}^\dagger\right)^2 \hat{a}^2.
\label{Eq:Hamiltonian_4_QuadraticKerr}
\end{equation}
The special feature of this case is that both the squeezing and Kerr terms involve products of exactly four creation or annihilation operators in each term. As a result, unlike the three cases considered above, we cannot say which term will be dominant in the infinite-photon-number limit independently of the value of $K$. Furthermore, if one of the two terms is dominant over the other, this dominance will be largely independent of $N$. There will be no reversal of roles between the squeezing and Kerr terms as the photon number increases. The simulation results for this case are plotted in Fig.~\ref{Fig:AveragePhotonNumberVsSqueezingParameterKerr4Quadratic}. The change in behaviour is abrupt in this case. When $K<2$, all the different simulation sizes produce different results. In this case, the Kerr term is unable to effectively regulate the dynamics, and there is no separate convergence for even and odd $N$. When $K>2$, the Kerr term is sufficiently confining, and the simulation results are independent of $N$. At the critical point $K=2$, we observe a peculiar situation: the results for $N=1000$ and $1001$ are close to each other, and similarly for $N=10^4$ and $10^4+1$, i.e.~the parity dependence is weak. However, the dynamics changes significantly when we change $N$ from $10^3$ to $10^4$, somewhat similarly to the unregularized case $K=0$ in Fig.~\ref{Fig:AveragePhotonNumberVsSqueezingParameter}.

\begin{figure}[h]
\includegraphics[width=8cm]{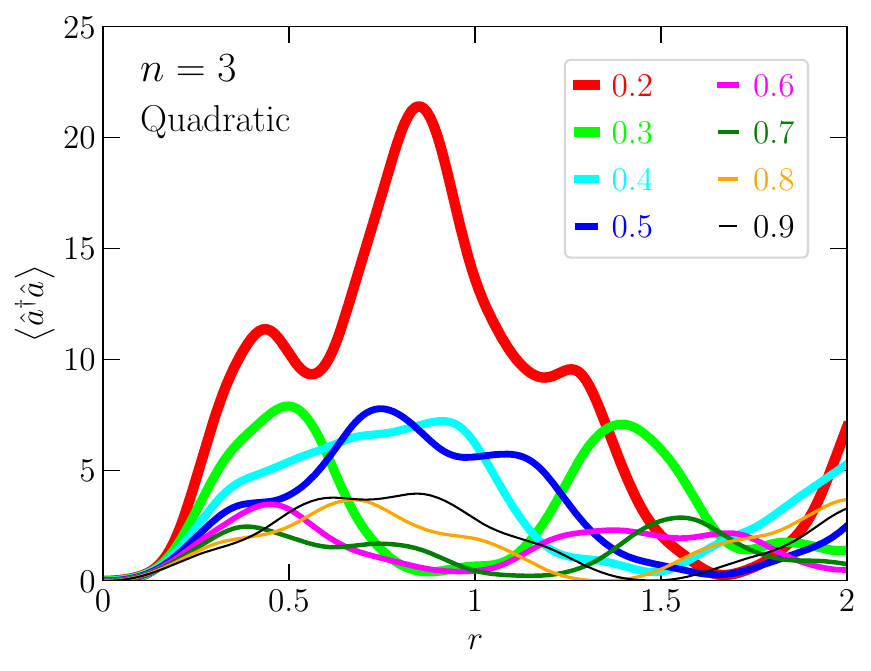}
\includegraphics[width=8cm]{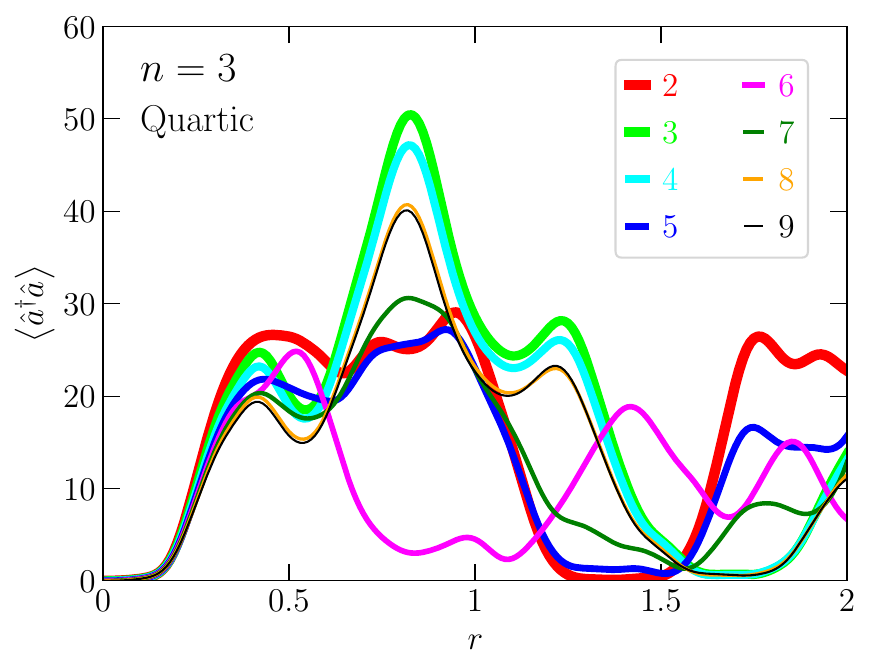}
\includegraphics[width=8cm]{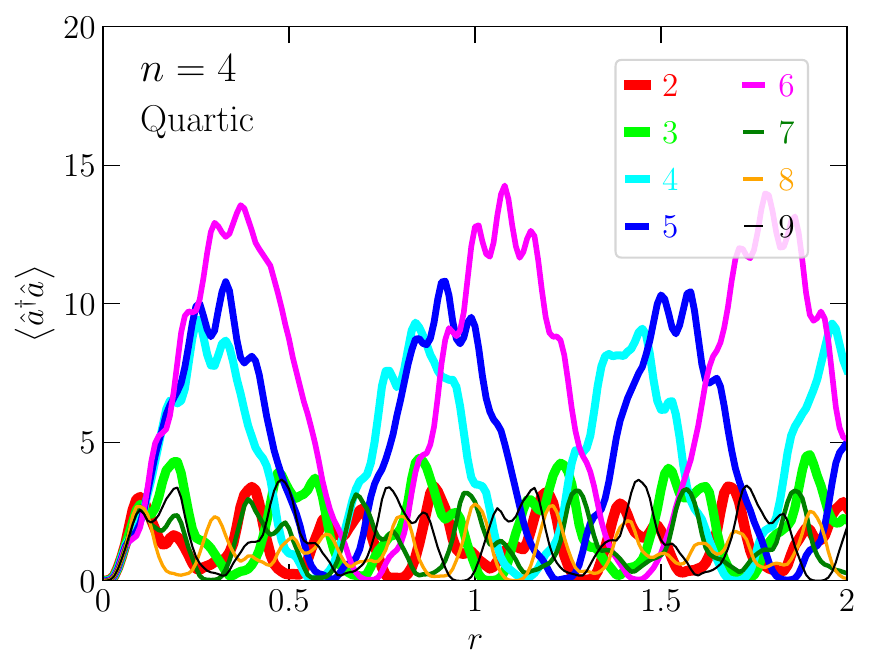}
\includegraphics[width=8cm]{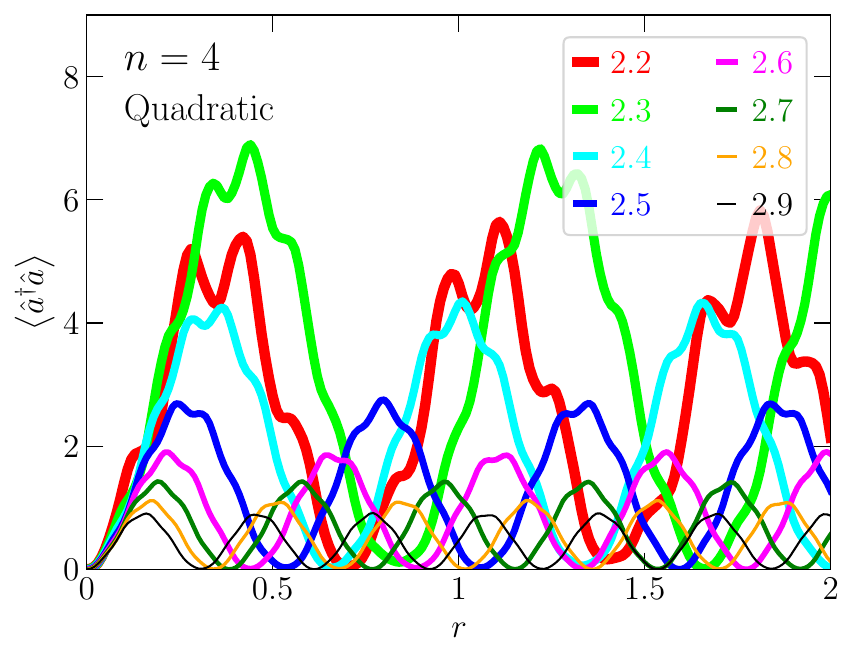}
\caption{Average photon number $\left\langle \hat{a}^{\dagger} \hat{a} \right\rangle$ for the state $\hat{U}_n(r)\ket{0}$ as a function of the squeezing parameter $r$ with a Kerr term of varying strength ($K$ or $K_4$). In all cases we use $N=1000$, keeping in mind that the Kerr term is strong enough that the results remain the same for any $N\geq 1000$ (see Figs.~\ref{Fig:AveragePhotonNumberVsSqueezingParameterKerr3Quadratic}-\ref{Fig:AveragePhotonNumberVsSqueezingParameterKerr4Quadratic}). For $n=3$ with a quadratic Kerr term, the different lines correspond to $K=0.2, 0.3, \cdots , 0.9$. For $n=3$ with a quartic Kerr term, the different lines correspond to $K_4=2 \times 10^{-6}, 3 \times 10^{-6}, \cdots , 9 \times 10^{-6}$. For $n=4$ with a quartic Kerr term, the different lines correspond to $K_4=2 \times 10^{-5}, 3 \times 10^{-5}, \cdots , 9 \times 10^{-5}$. For $n=4$ with a quadratic Kerr term, the different lines correspond to $K=2.2, 2.3, \cdots , 2.9$.}
\label{Fig:AveragePhotonNumberVsSqueezingParameterKerrVariableK}
\end{figure}

Taking into account that the dynamics represented in Figs.~\ref{Fig:AveragePhotonNumberVsSqueezingParameterKerr3Quadratic}, \ref{Fig:AveragePhotonNumberVsSqueezingParameterKerr3Quartic}, \ref{Fig:AveragePhotonNumberVsSqueezingParameterKerr4Quartic} and \ref{Fig:AveragePhotonNumberVsSqueezingParameterKerr4Quadratic} are generally different from those in Fig.~\ref{Fig:AveragePhotonNumberVsSqueezingParameter}, we perform a few additional calculations in which we vary the Kerr coefficient ($K$ or $K_4$) and plot the dynamics of multiple $K$ or $K_4$ values together. The results are shown in Fig.~\ref{Fig:AveragePhotonNumberVsSqueezingParameterKerrVariableK}. It should be noted that all the results plotted in Fig.~\ref{Fig:AveragePhotonNumberVsSqueezingParameterKerrVariableK} are independent of $N$, provided that $N\geq 1000$. In the majority of the cases, the dynamics are different from those obtained in Figs.~\ref{Fig:AveragePhotonNumberVsSqueezingParameterKerr3Quadratic}, \ref{Fig:AveragePhotonNumberVsSqueezingParameterKerr3Quartic}, \ref{Fig:AveragePhotonNumberVsSqueezingParameterKerr4Quartic} and \ref{Fig:AveragePhotonNumberVsSqueezingParameterKerr4Quadratic}. However, a few of the lines in Fig.~\ref{Fig:AveragePhotonNumberVsSqueezingParameterKerrVariableK} clearly resemble the even-$N$ lines in Fig.~\ref{Fig:AveragePhotonNumberVsSqueezingParameter}. A smaller number of lines in Fig.~\ref{Fig:AveragePhotonNumberVsSqueezingParameterKerrVariableK} generally resemble the odd-$N$ lines in Fig.~\ref{Fig:AveragePhotonNumberVsSqueezingParameter}. One special case is the one corresponding to $n=4$ with a quadratic Kerr term. As we move away from the critical point ($K=2$), the dynamics become increasingly regular, with both the oscillation amplitude and period decreasing with increasing $K$. Again, this phenomenon is mathematically demonstrated in Section \ref{Sec:Math}, where we show that $K=2$ precisely corresponds to the critical value of the Kerr parameter over which the Hamiltonian becomes essentially self-adjoint, and thus the limiting dynamics as $N\to\infty$ becomes unique.

Another interesting feature in Fig.~\ref{Fig:AveragePhotonNumberVsSqueezingParameterKerrVariableK} is the strong, irregular dependence of the oscillation pattern on the Kerr coefficient $K$ or $K_4$. This result suggests that it should be possible to extract the value of the Kerr coefficient by driving the oscillator at one of the multi-photon resonance frequencies and observing the ensuing dynamics. One obvious complication is that, in practice, there will not be a single Kerr term that accurately describes the nonlinearity up to $N=1000$.


\begin{thebibliography}{99}
\bibitem{Footnote} These two extensions correspond to the choices $U=\pm I$ in the general parametrization provided by \cite[Theorem 4.16]{FischerInPreparation}. The coefficients $d_j^{(i)}$ are solutions of a recurrence relation provided in the paper.
\end{thebibliography}
\end{document}